\documentclass[a4paper,UKenglish]{lipics-v2016}
\usepackage[utf8]{inputenc}
\usepackage{mathscinet}
\usepackage{comment}
\usepackage{graphicx}
\graphicspath{{figures/}}
\usepackage{booktabs}
\usepackage{hyperref}
\usepackage{pgfkeys}
\usepackage[numbers,sort&compress]{natbib}
\usepackage{enumerate}
\usepackage{verbatim}
\usepackage{nicefrac}
\usepackage{framed}
\usepackage{paralist}

\usepackage[dvipsnames]{xcolor}
\usepackage{tikz}
\usetikzlibrary{fit}
\usetikzlibrary{intersections}
\usetikzlibrary{decorations.pathreplacing}
\usetikzlibrary{calc}

\usepackage{amsmath}
\usepackage{amssymb}

\usepackage{xspace}
\usepackage{authblk}
\usepackage[color=green!20]{todonotes}

\theoremstyle{plain}
\newtheorem*{thm*}{Theorem}

\newtheorem*{eth-env}{Exponential Time Hypothesis}

\def\pure2dir{\textsc{Pure 2-Dir}\xspace}
\def\2dir{\textsc{2-Dir}\xspace}
\def\threedir{\textsc{3-Dir}\xspace}
\def\3sat{\textsc{3-Sat}\xspace}

\newcommand{\Coloring}[1]{\textsc{$#1$-Coloring}\xspace}
\newcommand{\LColoring}[1]{\textsc{List $#1$-Coloring}\xspace}

\def\mis{\textsc{Max Independent Set}\xspace}
\def\mds{\textsc{Min Dominating Set}\xspace}
\def\mids{\textsc{Min Independent Dominating Set}\xspace}
\def\mcds{\textsc{Min Connected Dominating Set}\xspace}
\def\mcli{\textsc{Max Clique}\xspace}
\def\mvc{\textsc{Min Vertex Cover}\xspace}
\def\fvs{\textsc{Min Feedback Vertex Set}\xspace}

\def\kColoring{$k$-\textsc{Coloring}\xspace}
\def\LkColoring{\textsc{List} $k$-\textsc{Coloring}\xspace}

\def\kfvs{\textsc{$k$-Feedback Vertex Set}\xspace}
\def\kpath{\textsc{$k$-Path}\xspace}
\def\kcycle{\textsc{$k$-Cycle}\xspace}

\def\ekcycle{\textsc{Exact $k$-Cycle}\xspace}

\renewcommand{\leq}{\leqslant}

\renewcommand{\geq}{\geqslant}

\title{Optimality Program in Segment and String Graphs}

\author[1]{\'Edouard Bonnet}
\author[2]{Pawe\l{} Rz\k{a}\.zewski}
\affil[1]{Department of Computer Science, Middlesex University, London}
\affil[2]{Faculty of Mathematics and Information Science,\\
		Warsaw University of Technology\\
\texttt{edouard.bonnet@dauphine.fr}, \texttt{p.rzazewski@mini.pw.edu.pl}}		
\authorrunning{\'E. Bonnet and P. Rz\k{a}\.zewski}
\Copyright{\'Edouard Bonnet and Pawe\l{} Rz\k{a}\.zewski}

\begin{document}

\maketitle

\begin{abstract}
  Planar graphs are known to allow subexponential algorithms running in time $2^{O(\sqrt n)}$ or $2^{O(\sqrt n \log n)}$ for most of the paradigmatic problems, while the brute-force time $2^{\Theta(n)}$ is very likely to be asymptotically best on general graphs. 
  Intrigued by an algorithm packing curves in $2^{O(n^{2/3}\log n)}$ by Fox and Pach [SODA'11], we investigate which problems have subexponential algorithms on the intersection graphs of curves (string graphs) or segments (segment intersection graphs) and which problems have no such algorithms under the ETH (Exponential Time Hypothesis).
  Among our results, we show that, quite surprisingly, \Coloring{3} can also be solved in time $2^{O(n^{2/3}\log^{O(1)}n)}$ on string graphs while an algorithm running in time $2^{o(n)}$ for \Coloring{4} even on axis-parallel segments (of unbounded length) would disprove the ETH.
  For \Coloring{4} of unit segments, we show a weaker ETH lower bound of $2^{o(n^{2/3})}$ which exploits the celebrated Erd\H{o}s-Szekeres theorem.
  The subexponential running time also carries over to \fvs but not to \mds and \mids. 
\end{abstract}

\section{Introduction}

Most combinatorial optimization and decision problems admit subexponential algorithms when restricted to planar graphs.
More precisely, they can be solved in time $2^{O(\sqrt n)}$, or $2^{\tilde{O}(\sqrt n)}$ on planar graphs with $n$ vertices, while under the ETH (Exponential Time Hypothesis, which asserts that \3sat cannot be solved in subexponential time \cite{ImpagliazzoPaturi,DBLP:journals/jcss/ImpagliazzoPZ01}) they do not admit an algorithm running in time $2^{o(n)}$ on general graphs.
The former is due to the facts that planar graphs have treewidth $O(\sqrt n)$ and that we have efficient algorithms parameterized by the treewidth $\text{tw}$ of the graph, namely running in $2^{O(\text{tw})}n^{O(1)}$, or $2^{\tilde{O}(\text{tw})}n^{O(1)}$.

The so-called bidimensionality theory \cite{DemaineH04,DemaineFHT04,DemaineH08,DemaineH08a} pushes this speed-up further by yielding $2^{O(\sqrt k)}n^{O(1)}$ algorithms where $k$ is the targeted size of a solution (think for example of the problems of finding a maximum independent set or a minimum dominating set of size $k$).
In a nutshell, it exploits a deep structural result by Robertson, Seymour, and Thomas \cite{DBLP:journals/jct/RobertsonST94}: planar graphs with treewidth $\text{tw}$ have a $\Theta(\text{tw})$-by-$\Theta(\text{tw})$ grid as a minor (i.e., any graph obtained by deleting vertices and edges, and contracting edges).
Thus, if the presence of a large grid minor makes the problem trivial (as in, one can always answer yes or always answer no), then one only has to solve efficiently instances with low treewidth; which, as we noted, can often be done.
The claimed running time is obtained by defining large grids as $\Theta(\sqrt k)$-by-$\Theta(\sqrt k)$, since their absence as minors imply that the treewidth is in $O(\sqrt k)$.
The bidimensionality theory is extremely versatile.
It also gives approximation schemes and linear vertex kernels and could be generalized to graphs with bounded genus and graphs excluding a fixed minor \cite{DemaineFHT05b}.

A natural line of research is to generalize or extend the subexponential (parameterized) algorithms to classes of graphs which do not fall into those categories. For geometric intersection graphs, the situation is much richer than for planar graphs.
For instance, Marx and Pilipczuk already observed that packing problems (of the kind of \mis) are more broadly subject to subexponential algorithms -- running typically in $n^{O(\sqrt k)}$ -- than covering problems (of the kind of \mds) -- for which $n^{O(k)}$ is essentially optimal under the ETH \cite{MarxP15,MarxP15a}.

We briefly survey the existing results in the design of subexponential algorithms on geometric intersection graphs.
A prominent role is played by intersection graph of families of {\em fat objects}, i.e., objects for which the aspect ratio (their length divided by their width) is bounded.
We highlight that fat objects, and in particular disks and squares, often allow faster algorithms and the so-called square-root phenomenon.
As we will see, subexponential algorithms are less frequent on intersection graphs of curves and segments but nevertheless present such as exemplified by \mis, \Coloring{3}, and \fvs.

\paragraph*{Subexponential algorithms on geometric intersection graphs.}
By a {\em ply} of a family of geometric objects we denote the maximum number of objects covering a single point.
Smith and Wormald show that for any collection of $n$ convex fat objects with ply $p$ there is a balanced separator of size $O(\sqrt{n p})$ \cite{SmithW98}. This leads to subexponential algorithms when the ply is constant, or in general for problems becoming trivial when the ply is too large, such as \Coloring{k}.
The $2^{\tilde{O}(\sqrt{n k})}$-time algorithm that this win-win provides for coloring $n$ fat objects, say disks, with $k$ colors is shown essentially optimal under the ETH by Bir\'{o} et al. \cite{Biro17}.

A next step may consist of designing FPT\footnote{with running time $f(k)n^{O(1)}$} or XP\footnote{with running time $n^{f(k)}$} algorithms where the dependency in the parameter is subexponential (for problems of the form ,,find $k$ vertices such that...``).
Using a shifting argument à la Baker \cite{Baker}, Alber and Fiala obtain a $n^{O(\sqrt k)}$-time to decide if one can find $k$ disjoint unit disks or squares among $n$ \cite{AlberF04}.
Marx and Pilipczuk generalize this result to packing $k$ disjoint polygons among $n$ in the same time \cite{MarxP15,MarxP15a}.
Their approach is based on guessing a small separator in the medial axis (i.e., the Voronoi diagram of polygons) of a supposed solution, as suggested by Adamaszek and Wiese and Har-Peled to obtain QPTAS for geometric packing problems \cite{AdamaszekW14,HarPeled14,AdamaszekHW17}.

Marx showed that \mis and \mds in the intersection graphs of disks or squares are W[1]-complete, and therefore unlikely to be FPT \cite{Marx06}.
Those reductions also show that the $n^{O(\sqrt k)}$ algorithms \cite{MarxP15,MarxP15a} are essentially optimal under the ETH. 
Fomin et al. \cite{FominLS12} observed that unit disks of bounded degree have treewidth $O(\sqrt n)$ and used this fact to extend bidimensionality to unit disk graphs for a handful of problems.  
Recently, a superset of the previous authors gave $2^{O(\sqrt k)}n^{O(1)}$-time algorithms for \kfvs, \kpath, \kcycle, \ekcycle \cite{FominLPSZ17}.
%
%
\paragraph*{Non-fat objects: segments and strings.}

Segment intersection graphs (or segment graphs in short) are the intersection graphs of straight-line segments in the plane.
They are called unit segments if all the segments of a representation share the same length.
For a fixed integer  $k$, $k$-{\sc Dir} is defined as the set of intersection graphs of segments, each parallel to one of fixed $k$ directions. Strings graphs are the intersection graphs of simple curves in the plane.
Those curves can be assumed polygonal without loss of generality.
The vertices of the polygonal curves in a geometric representation are called \emph{geometric vertices} not to confuse them with the actual vertices of the graph. 
As shown by Kratochv\'{i}l and Matou\v{s}ek, there are string graphs with $n$ vertices, which require $2^{\Omega(n)}$ geometric vertices in any string representation with polygonal curves~\cite{KratochvilM91}.

A systematic study of segment graphs and their subclasses was initiated by Kratoch\'{i}l and Matou\v{s}ek~\cite{KRATOCHVIL1994289}. It is interesting to point out that every planar graph is a segment graph, as shown by Chalopin and Gon\c{c}alves~\cite{ChalopinG09} (this was a long-standing conjecture by Scheinerman~\cite{Scheinerman}).

The class of string graphs is very general, as it includes split graphs (i.e., graphs whose vertices can be partitioned into two sets inducing a clique and an independent set), intersection graphs of bodies (i.e., compact shapes with non-empty interior), or incomparability graphs (i.e., graphs whose vertex set is given by the set of elements of a poset, and edges join elements that are incomparable).

Bir\'{o} et al. showed that even though coloring disks or more generally fat objects with a constant number of colors can be solved in $2^{\tilde{O}(\sqrt n)}$ \cite{Biro17}, 6-coloring axis-parallel segments (\2dir) in time $2^{o(n)}$ would refute the ETH. This suggests that subexponential algorithms are less frequent on the intersection graphs of non-fat objects such as segments and strings. On the other hand, Fox and Pach presented a subexponential algorithm for \mis on string graphs \cite{FoxP11}. Their approach uses a win-win strategy and is based on the existence of balanced separators in string graphs. 
Fox, Pach, and T\'oth showed that string graphs with $m$ edges have balanced separators of size $O(m^{3/4} \; \log m)$, and conjectured that there is always a separator of size $O(\sqrt{m})$ \cite{DBLP:journals/jct/FoxPT10}. Matou\v{s}ek showed that string graphs admit a balanced separator of size $O(\sqrt{m} \; \log m)$~\cite{DBLP:journals/cpc/Matousek14}. Finally, very recently Lee improved the result of Matou\v{s}ek, proving the conjecture.
\begin{theorem}[Lee \cite{Lee16}]\label{thm-stringsep}
Every string graph with $m$ edges has a balanced separator of size $O(\sqrt{m})$. Moreover, it can be found in polynomial time, provided that the geometric representation is given.
\end{theorem}
Let us point out that this result generalizes the famous planar separator theorem by Lipton and Tarjan \cite{DBLP:journals/siamcomp/LiptonT80}, as planar graphs are string graphs and the number of edges in a planar graph is linear in the number of vertices. This also shows that Theorem \ref{thm-stringsep} is best possible (up to the constants), as the planar separator theorem is asymptotically tight.
\paragraph*{Our contributions.}
We show that the subexponential algorithm for \mis in string graphs by Fox and Pach \cite{FoxP11}, running in time  $2^{\tilde{O}(n^{2/3})}$, can be extended to \Coloring{3} and \fvs.
As in the algorithm of Fox and Pach, the central idea is a win-win: either the graph is rather sparse and the separator of Theorem~\ref{thm-stringsep} gives a speed-up, or the graph has a high-degree vertex (used for \Coloring{3}) or a large biclique (used for \fvs) and an efficient branching can be performed.
Refining a lower bound of Biro et al.~\cite{Biro17}, we complement this former result by showing that for any $k\geq 4$, \kColoring cannot be solved in $2^{o(n)}$ even on axis-parallel segments, unless the ETH fails.
The reduction relies on having segment lengths with two different orders of magnitude.
We therefore ask if unit segments could allow a faster algorithm for \kColoring for $k\geq 4$.
Under the ETH, we provide a stronger lower bound than the one for planar graphs (which refutes a running time $2^{o(\sqrt n)}$) and show that unit segments cannot be $k$-colored in $2^{o(n^{2/3})}$ for any $k \geq 4$.
Our construction uses the fact, closely related to the famous Erd\H{o}s-Szekeres~\cite{Erdos1987} theorem, that any permutation on $n$ totally ordered elements can be partitioned into $O(\sqrt n)$ monotone subsequences (see Brandst\"{a}dt and Kratsch~\cite{BrandstadtK86}).

We then give tight ETH lower bounds for \textsc{Min (Connected) Dominating Set} and \mids on segments and \mcli on strings.
For that, we design reductions whose number $n$ of produced segments is linear in $N+M$ from satisfiability problems with $N$ variables and $M$ clauses.
Indeed, the sparsification lemma of Impagliazzo et al. \cite{Sparsification} implies that those satisfiability problems are not solvable in $2^{o(N+M)}$ unless the ETH fails; which enables us to conclude that the problems are not solvable in $2^{o(n)}$ under the ETH, on graphs with $n$ vertices.

Although the NP-hardness of the aforementioned problems is known for segment intersection graphs \cite{ZverovichZ95,CabelloCL13}, getting such linear reductions might be difficult.

For instance, while it is known that planar graphs are a subclass of segment intersection graphs \cite{ChalopinG09}, implying the NP-hardness of all the problems of Table~\ref{tab:results-table} except \Coloring{k} for $k \geqslant 4$ and \mcli, this fact does not serve our purpose since they can be solved in time $2^{O(\sqrt n)}$ on planar graphs.
The situation is an interesting intermediate between planar and general graphs.
Our objects \emph{can} intersect but we cannot afford crossover gadgets (at least not quadratically many). 
Certain intersections create unwanted edges of which we have to tame the importance.
It is also noteworthy that segment/string graphs cannot be expanders since if they have constant degree, by Theorem~\ref{thm-stringsep}, they have treewidth $\tilde{O}(\sqrt n)$.
Hence, we are deprived of the \emph{usual hardest instances}.

\begin{table}[h!]
\centering
\begin{tabular}{l  l l}
\toprule
Problem & Upper bound & Lower bound \\
\midrule
\mis & $2^{\tilde{O}(\sqrt n)}p^{O(1)}$, $2^{\tilde{O}(n^{2/3})}$ & $2^{o(\sqrt n)}$ \\
\Coloring{3} & $\mathbf{2^{\tilde{O}(n^{2/3})}}$ & $2^{o(\sqrt n)}$ \\
\kColoring for every $k \geq 4$ & $2^{O(n)}$ & $\mathbf{2^{o(n)}}$ ~~~(even in \2dir) \\
\kColoring for every $k \geq 4$ & $2^{O(n)}$ & $\mathbf{2^{o(n^{2/3})}}$ in unit \threedir \\
\fvs & $\mathbf{2^{\tilde{O}(n^{2/3})}}$ & $2^{o(\sqrt n)}$ \\
\textsc{\footnotesize{Min} \small{(Connected)}\footnotesize{} Dominating Set} & $2^{O(n)}$ & $\mathbf{2^{o(n)}}$  \\
\mids & $2^{O(n)}$ & $\mathbf{2^{o(n)}}$ \\
\mcli & $2^{O(n)}$ & $\mathbf{2^{o(n)}}$ \\
\bottomrule
\end{tabular}
\caption{Upper and lower bounds for classical problems on string and segment graphs.
The \textbf{upper bounds} work on \textbf{string graphs}.
The \textbf{lower bounds} are designed on \textbf{segment graphs}, unless precised otherwise.
Our results are written in bold.
By $p$ we denote the number of geometric vertices if a geometric representation is given.}
\label{tab:results-table}
\end{table}
\paragraph*{Geometric representation and robust algorithms.}
In case of graphs with geometric representations, it is important to distinguish between a geometric intersection graph (i.e., a pure abstract structure, for which we know that some geometric representation exists), and the representation itself.
Note that this is not the case with planar graphs, as finding a plane embedding can be done in linear time \cite{DBLP:journals/jgaa/BoyerM04}.

Finding a segment or string representation of a graph was shown to be NP-hard by Kratochv\'{i}l \cite{DBLP:journals/jct/Kratochvil91a}, and Kratochv\'{i}l and Matou\v{s}ek  \cite{KRATOCHVIL1994289}, respectively. However, it was very unclear if the problems are in NP (which is usually the trivial part of an NP-completeness proof).
As mentioned above, Kratochv\'{i}l and Matou\v{s}ek \cite{KratochvilM91} showed that some string graphs require a representation of exponential size, which proved that the simple idea of exhaustively guessing the representation cannot work for this problem. Finally, the NP-membership of recognizing string graphs was proven by Schaefer, Sedgwick, and \v{S}tefankovi\v{c} \cite{DBLP:journals/jcss/SchaeferSS03}.

The story of recognizing segment graphs is even more interesting. On the first sight, the situation seems simpler than for strings, as the number of geometric points in a segment representation is clearly polynomial in $n$. However, it appears that there are segment graphs, whose every segment representation requires points with coordinates doubly exponential in $n$, i.e., using $2^{\Omega(n)}$ digits (see Kratochv\'{i}l and Matou\v{s}ek \cite{KRATOCHVIL1994289}, and McDiarmid and M\"{u}ller \cite{DBLP:journals/jct/McDiarmidM13}). Finally, the problem was shown to be complete for the class $\exists \mathbb{R}$~
(see Schaefer and \v{S}tefankovi\v{c}~\cite{Schaefer2017}), i.e., the class of problems reducible in polynomial time to deciding if a given existential formula over the reals is true. This is a strong evidence that the problem is not in NP. For a very nice exposition of the $\exists \mathbb{R}$-completeness proof, see Matou\v{s}ek \cite{DBLP:journals/corr/Matousek14}.

All this shows that a requirement of an explicit geometric representation of an input graph may be a serious drawback of an algorithm. We call an algorithm {\em robust} if it takes only an abstract structure as an input, and either computes the solution, or concludes (correctly) that the input graph does not belong to the desired class.
On the one hand, our algorithms (see Section \ref{sec:upper}) are robust, but work slightly faster if the input is given along with the geometric representation. On the other hand, the lower bounds (see Section \ref{sec:lower}) hold even if the geometric representation is given explicitly.

\section{Upper bounds} \label{sec:upper}

Fox and Pach showed that, on string graphs, a maximum independent set can be computed in subexponential time:

\begin{theorem}[Fox \& Pach \cite{FoxP11}]
\mis can be solved in time $2^{O(n^{2/3} \log n)}$ in string graphs with $n$ vertices.
\end{theorem}

In their paper, they give a worse running time than the one claimed above.
This is because they used the $O(m^{3/4} \log m)$ separator theorem \cite{DBLP:journals/jct/FoxPT10}, which has been recently improved to $O(\sqrt{m})$ \cite{Lee16}.
The algorithm is a simple win-win argument.
If there is a vertex with degree at least $n^{1/3}$, then either removing it or selecting it and removing its neighbors gives a branching 
$F(n) \leq F(n-1) + F(n-\lceil n^{1/3} \rceil-1)$.
Otherwise, if all the vertices have degree smaller than $n^{1/3}$, the graph is rather sparse and the balanced separator of size $O(\sqrt m)=O(n^{2/3})$ provides an efficient divide-and-conquer.
The threshold $n^{1/3}$ is computed so that it balances the running time of those two subroutines and gives the claimed overall asymptotic time. 

This result was somewhat improved by Marx and Pilipczuk \cite{MarxP15,MarxP15a} based on an approach introduced by Adamaszek, Har-Peled, and Wiese \cite{AdamaszekHW17} to get QPTAS for geometric problems.
However, their algorithm necessitates that the string graph is given with a representation by polygonal curves on a polynomial number of geometric vertices.
\begin{theorem}[Marx \& Pilipczuk \cite{MarxP15}]
\mis can be solved in time $2^{O(\sqrt n \log n)}p^{O(1)}$ in string graphs with $n$ vertices, where the strings are given as polygonal curves on a total of $p$ geometric vertices.
\end{theorem}
In a nutshell, the idea is to exhaustively guess a small balanced face-separator in the Voronoi diagram of a supposed (although not known) fixed solution, and solve recursively the two subinstances in the inside and outside of this separator.

If this approach does not seem to generalize easily to coloring problems, the win-win of Fox and Pach can be transported to \Coloring{3} with a bit more arguments.

\begin{theorem} \label{thm-3col-string}
\Coloring{3} (even \LColoring{3}) of a string graph with $n$ vertices can be decided in time:
\begin{compactitem}
\item  $2^{O(n^{2/3} \log n)}$, if the geometric representation is given,
\item $2^{O(n^{2/3} \log^2 n)}$, even without geometric representation.
\end{compactitem} 
\end{theorem}

\begin{proof}
  Consider an instance $(G,L)$ of \LColoring{3} with $n$ vertices (in \LColoring{k} each vertex $v$ is equipped with a list $L(v) \subseteq [k]$ and we want to find a proper coloring, in which every vertex receives a color from its list).
  Observe that without loss of generality we can assume that each list has two or three elements.
  Indeed, if there is a vertex with just one allowed color, we can fix this color and remove it from the list of each of its neighbors.
  Let $N$ be the sum of the lengths of the lists; clearly $2n \leq N \leq 3n$.

First, assume that $G$ has no vertex with degree larger than $n^{1/3}$, then the number $m$ of edges is $O(n^{4/3})$.
By Theorem \ref{thm-stringsep}, $G$ has a balanced separator of size $O(\sqrt{m}) = O(n^{2/3})$.
We can find this separator in polynomial time, if the representation is given, or by exhaustive guessing in time $n^{O(n^{2/3})}=2^{O(n^{2/3} \log n)}$, without using a representation.
Then we list all possible colorings of the separator and proceed with a standard divide-and-conquer approach. The depth of the recursion is $O(\log n)$, so the total time complexity of this step is $2^{O(n^{2/3} \log^2 n)}$, or $2^{O(n^{2/3} \log n)}$ if we use the geometric representation to find a separator.

If there is a vertex $v$ of degree at least $n^{1/3}$, then one among the lists: $\{1,2\}, \{1,3\}, \{2,3\},$ $\{1,2,3\}$ appears on at least $n^{1/3}/4$ of its neighbors.
Thus there are two colors (say, $1$ and $2$) that appear in lists of at least $n^{1/3}/4$ of neighbors of $n$.
Since the list of $v$ has size at least two, one of these colors (say $1$) appears on the list of $v$.
We branch into two possibilities: choosing the color 1 for $v$ (then we exclude $1$ from the lists of all neighbor of $v$), and not choosing 1 for $v$ (then we remove 1 from the list of $v$).
The complexity $F$ of this step is given by the recursion $F(N) \leq  F(N - n^{1/3}/4) + F(N-1) \leq F(N - N^{1/3}/(3^{1/3} \cdot 4)) + F(N-1)$.
This inequality is satisfied by $F(N) = 2^{O(N^{2/3} \log N)} = 2^{O(n^{2/3} \cdot \log n)}$.

Combining these two cases gives the claimed time complexity.
Finally, observe that if the input graph is not a string graph, then the exhaustive search for a separator might fail, and then we can report a wrong input instance.
\end{proof}

For \fvs, there is no obvious subexponential branching on a high-degree vertex.
Instead, we use the following theorem by Lee.
\begin{theorem}[Lee \cite{Lee16}] \label{thm-string-Ktt}
 There is a constant $c$ such that for any $t \geqslant 1$, $K_{t,t}$-free string graphs on $n$ vertices have fewer than $c \cdot t\log t \cdot n$ edges. 
\end{theorem}

It is worth mentioning that Fox and Pach \cite[Theorem 5]{FoxPachArxiv} obtained a slightly weaker result  with $\log^{O(1)}t$ instead $\log t$.

\begin{theorem} \label{thm-fvs-string}
\fvs on string graphs with $n$ vertices can be solved in time
\begin{compactitem}
\item  $2^{O(n^{2/3} \log^{3/2} n)}$, if the geometric representation is given,
\item $2^{O(n^{2/3} \log^{5/2} n)}$, even without geometric representation.
\end{compactitem} 
\end{theorem}

\begin{proof}
The proof is similar to the proof of Theorem \ref{thm-3col-string}. Let $G$ be the input graph.
If $G$ has fewer than $c/3 \cdot  n^{4/3} \log n$ edges (where $c$ is a constant from Theorem \ref{thm-string-Ktt}), then by Theorem \ref{thm-stringsep} there is a balanced separator of size $O(\sqrt m) = O(n^{2/3} \log^{1/2} n)$, and a divide-and-conquer approach yields a running time $2^{O(n^{2/3} \log^{3/2} n)}$ (with the representation, we use it to find the separator), or $2^{O(n^{2/3} \log^{5/2} n)}$ (without the representation, we guess the separator exhaustively).

Otherwise, by Theorem \ref{thm-string-Ktt}, there is a subgraph of $G$ isomorphic to the biclique $K_{n^{1/3},n^{1/3}}$. We can find it by exhaustive guessing in time $n^{n^{2/3}} \cdot poly(n) = 2^{O(n^{2/3} \log n)}$.
Observe that any feedback vertex set of $G$ must contain all but one vertex of one bipartition class of the biclique.
Guessing which vertex is {\emph not} chosen into the solution gives us a branching algorithm, whose complexity is given by the recursion $F(n) \leq 2^{O(n^{2/3} \log n)} + 2n^{1/3}F(n-n^{1/3}+1)$, which is solved by $F(n) = 2^{O(n^{2/3} \log n)}$.
If the exhaustive search for a separator or a biclique fails, then we can correctly report that the input graph is not a string graph.
\end{proof}

\section{Lower bounds} \label{sec:lower}
Rather surprisingly, the win-win for \Coloring{3} abruptly ceases to work for \kColoring for every $k \geq 4$.
First, let us consider the \LColoring{4}.
Following Kratochv\'il and Matou\v{s}ek \cite{KRATOCHVIL1994289}, by \pure2dir we denote graphs admitting a \2dir representation in which parallel segments do not intersect. Observe that such a graph is bipartite.

\begin{theorem} \label{thm:2dir-lists}
\LColoring{4} of a \pure2dir graph cannot be solved in time $2^{o(n)}$, unless the ETH fails.
\end{theorem}

\begin{proof}
Let $\Phi$ be a \3sat formula with $n$ variables $v_1,v_2,\ldots,v_n$ and $m$ clauses $C_1,C_2,\ldots,C_m$.
By repeating some literals in a clause, we may assume that each clause contains exactly three literals.
For a clause $C_i$, let $v^i_1,v^i_2,v^i_3$ denote the variables of $C_i$.

We construct a \2dir graph $G$ with lists $L$ of colors from the set $\{1,2,3,4\}$, such that $\Phi$ is satisfiable if and only if $G$ is list-colorable with respect to the lists $L$.

For each variable $v_i$, we introduce a horizontal segment called $x_i$. For each clause $C_i$ we introduce three vertical segments $y^i_1,y^i_2,y^i_3$, corresponding to $v^i_1,v^i_2$, and $v^i_3$, respectively.
We arrange them in a grid-like way (see Figure \ref{fig-2dir}). 
One may observe that the intersection graph induced by those segments is a biclique.
We set the lists of each $x_i$ to $\{1,2\}$ and the lists of each $y^i_1,y^i_2,y^i_3$ to $\{3,4\}$.
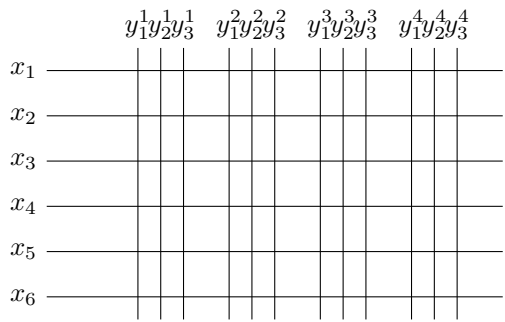
\begin{figure}[h]
\centering
\begin{tikzpicture}[scale=0.6]
\newcommand*\nn{6}
\newcommand*\mm{4}

\foreach \i in {1,...,\nn}
{	
	\draw (2*\mm+2,-\i)--(0,-\i) node [left] {$x_{\i}$};	
}
\foreach \i in {1,...,\mm}
{	
	\draw (2*\i,-\nn-0.5)--(2*\i,-0.5) node [above] {$y^{\i}_1$};	
	\draw (2*\i+0.5,-\nn-0.5)--(2*\i+0.5,-0.5) node [above] {$y^{\i}_2$};	
	\draw (2*\i+1,-\nn-0.5)--(2*\i+1,-0.5) node [above] {$y^{\i}_3$};	
}

\end{tikzpicture}
\caption{The arrangement of variable- and occurrence-segments in $G$.}
\label{fig-2dir}
\end{figure}

The colors 1 and 2 used for coloring $x_i$ will be interpreted, respectively, as {\em true} and {\em false} values given to $v_i$, while the colors 3 and 4 given to $y^i_j$ will be interpreted, respectively, as {\em true} and {\em false} values given to the literal corresponding to $v^i_j$.

To ensure this, we need to introduce {\em equality gadgets} and {\em inequality gadgets}.
If the variable $v_i$ appears positively in the clause $C_j$ as its $\ell$-th literal, then at the crossing point of $x_i$ and $y^j_\ell$ we put the equality gadget ensuring that in any feasible coloring of $G$, the color of $x_i$ is 1 (2, respectively) if and only if the color of $y^i_\ell$ is 3 (4, respectively).
On the other hand, if $v_i$ appears negatively in $C_j$ as its $\ell$-th literal, then at the crossing point of $x_i$ and $y^j_\ell$ we put the inequality gadget ensuring that in any feasible coloring of $G$, the color of $x_i$ is 1 (2, respectively) if and only if the color of $y^j_\ell$ is 4 (3, respectively).

The equality gadget consists of 3 segments, arranged as depicted on Figure \ref{fig-in-equality}. Consider the equality gadget (left lists on Figure \ref{fig-in-equality}) and suppose $x_i$ gets the color 1. Then $a$ receives color 3, and $c$ gets the color 4. Thus the only choice for the color for $y^j_\ell$ is 3. The coloring can be extended by coloring $b$ to 2. The other cases are symmetric.
The inequality gadget is analogous and uses the right lists on Figure \ref{fig-in-equality}.  

\begin{figure}[h]
\begin{minipage}{.5\linewidth}%
\begin{tikzpicture}[scale=0.6]
\def\s{0.75}

\coordinate (xb) at (3,0) ;
\coordinate (xe) at (-5,0) ;
\coordinate (yb) at (0,-0.75) ;
\coordinate (ye) at (0,3.5) ;

\draw (xb) -- (xe) node [above] {$x_i$};
\draw (yb) -- (ye) node [left] {$y^j_\ell$};

\draw[dashed] (xe) --++(-\s,0) ;
\draw[dashed] (yb) --++(0,-\s) ;
\draw[dashed] (xb) --++(\s,0) ;
\draw[dashed] (ye) --++(0,\s) ;

\draw (-3,3) -- (-3,-1) node [below] {$a$};
\draw (-2.5,3) -- (-2.5,-1) node [below] {$b$};
\draw (-3.5,2.6) -- (0.5,2.6) node [right] {$c$};
\end{tikzpicture} 
\end{minipage}%
\begin{minipage}{.3\linewidth}
{\small
\begin{tabular}{l | l}
vertex & list \\ \hline
$x_i$ & 1,2 \\
$y^j_\ell$ & 3,4 \\ \hline
$a$ & 1,3 \\
$b$ & 2,4 \\
$c$ & 3,4 \\
\end{tabular}

\medskip
Lists in the\\equality gadget.
}
\end{minipage}%
\begin{minipage}{.3\linewidth}
{\small
\begin{tabular}{l | l}
vertex & list \\ \hline
$x_i$ & 1,2 \\
$y^j_\ell$ & 3,4 \\ \hline
$a$ & 1,4 \\
$b$ & 2,3 \\
$c$ & 3,4 \\
\end{tabular}

\medskip
Lists in the\\inequality gadget.
}
\end{minipage}
\caption{Equality and inequality gadgets. The arrangement of segments is the same in both gadget, the only difference is the lists.}
\label{fig-in-equality}
\end{figure}
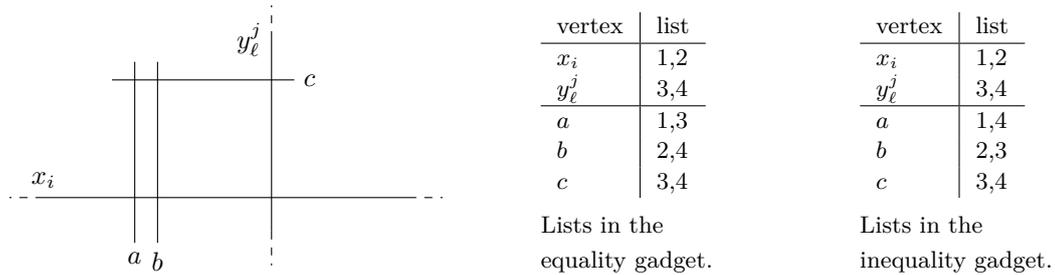

The only thing left is to ensure that the coloring of $y^j_1,y^j_2,y^j_3$ exists if and only if $C_j$ is satisfied. This is ensured by the {\em satisfiability gadget} depicted in Figure \ref{fig-sat}, attached to the top ends of $y_j^1,y^j_2$, and $y^j_3$. Note that the gadget can be colored if and only if one of  $y^j_1,y^j_2,y^j_3$ gets color 3, which is equivalent to one literal of $C_j$ being set to {\em true}.

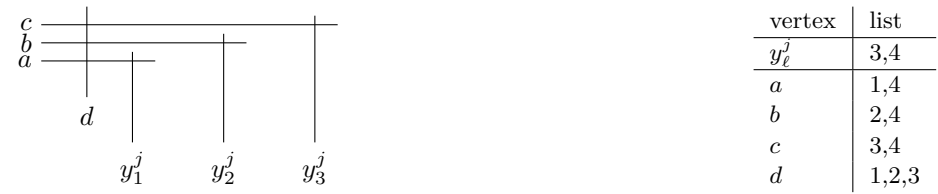
\begin{figure}[h!]
\begin{minipage}{.7\linewidth}%
\begin{tikzpicture}[scale=0.6]

\draw (0,3) -- (0,1) node [below] {$y^j_1$};
\draw (2,3.4) -- (2,1) node [below] {$y^j_2$};
\draw (4,3.8) -- (4,1) node [below] {$y^j_3$};

\draw (0.5,2.8) -- (-2,2.8) node [left] {$a$};
\draw (2.5,3.2) -- (-2,3.2) node [left] {$b$};
\draw (4.5,3.6) -- (-2,3.6) node [left] {$c$};

\draw (-1,4) -- (-1,2) node [below] {$d$};

\end{tikzpicture} 
\end{minipage}%
\begin{minipage}{.3\linewidth}
{\small
\begin{tabular}{l | l}
vertex & list \\ \hline
$y^j_\ell$ & 3,4 \\ \hline
$a$ & 1,4 \\
$b$ & 2,4 \\
$c$ & 3,4 \\
$d$ & 1,2,3
\end{tabular}
}
\end{minipage}
\caption{Satisfiability gadget.}
\label{fig-sat}
\end{figure}

The number of vertices of $G$ is 
$
n' = \underbrace{n}_{x_i} + \underbrace{3m}_{y^j_\ell} + \underbrace{9m}_{\text{(in)equality}} + \underbrace{4m}_{\text{satisfiability}} = \Theta(n+m).
$
On the other hand, a algorithm solving list coloring of $G$ in time $2^{o(n')}$ can be used to decide the satisfiability of $\Phi$ in time $2^{o(n')} = 2^{o(n+m)}$, which in turn contradicts the~ETH.
\end{proof}

The non-list version is obtained analogously to the hardness for \Coloring{6} in \cite{Biro17}.
 We include the proof to make the paper self-contained.

\begin{theorem} \label{thm:2dir}
For every fixed $k \geq 4$, the \kColoring problem of a \2dir graph cannot be solved in time $2^{o(n)}$, unless the ETH fails.
\end{theorem}
\begin{proof}
We modify the construction from the proof of Theorem \ref{thm:2dir-lists}.
We first introduce $k$ overlapping segments $R_1,R_2,\ldots,R_k$, whose coloring will serve as a reference coloring. 
Since these segments are pairwise intersecting, each of them receives a different color. 
We will denote by $i \in [k]$ the color assigned to $R_i$.

Now, for each segment $v$ of $G$, we want to simulate the list $L(v)$ from the instance of \LColoring{4} constructed in the proof of Theorem \ref{thm:2dir-lists}.
For every color $i \notin L(v)$, we want to introduce a segment $s_i$ intersecting $v$, which will always receive color $i$.

To achieve this, we first need to transport the reference coloring to every gadget. We split it into two parts -- we will separately transport colors 1 and 2, and colors greater than 2. 
The overall high-level idea is depicted in Figure \ref{fig:2dir-transport}. 
Observe that this already simulates the lists for every $x_i$.

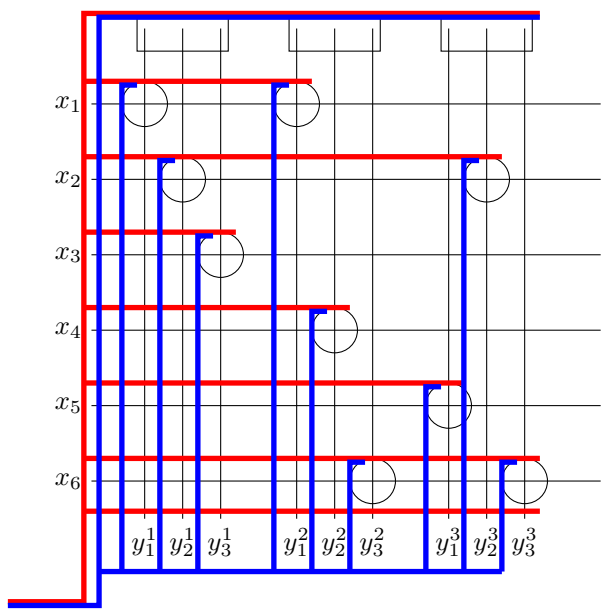
\begin{figure}[h!]
\begin{center}
\begin{tikzpicture}[scale=1]
\newcommand*\nn{6}
\newcommand*\mm{3}

\foreach \i in {1,...,\nn}
{	
	\draw (2*\mm+1,\nn-\i+1)--++(-2*\mm-0.7,0) node [left] {$x_{\i}$};	
}

\foreach \i in {1,...,\mm}
{	
	\draw (2*\i-1,\nn+1)--++(0,-\nn-0.5) node [below] {$y^{\i}_1$};	
	\draw (2*\i-0.5,\nn+1)--++(0,-\nn-0.5) node [below] {$y^{\i}_2$};	
	\draw (2*\i,\nn+1)--++(0,-\nn-0.5) node [below] {$y^{\i}_3$};		
	\draw (2*\i-1-0.1,\nn+0.7) rectangle ++(1.2,0.5);	
}

\draw (1, 6) circle (0.3);	
\draw (1.5, 5) circle (0.3);	
\draw (2, 4) circle (0.3);	

\draw (3, 6) circle (0.3);	
\draw (3.5, 3) circle (0.3);	
\draw (4, 1) circle (0.3);	

\draw (5, 2) circle (0.3);	
\draw (5.5, 5) circle (0.3);	
\draw (6, 1) circle (0.3);	

\draw[line width = 2, color = red] (6.2,7.2) --++(-6,0)--++(0,-7.8)--++(-1,0);
\draw[line width = 2, color = red] (0.2,6.3) --++(3,0);
\draw[line width = 2, color = red] (0.2,5.3) --++(5.5,0);
\draw[line width = 2, color = red] (0.2,4.3) --++(2,0);
\draw[line width = 2, color = red] (0.2,3.3) --++(3.5,0);
\draw[line width = 2, color = red] (0.2,2.3) --++(5,0);
\draw[line width = 2, color = red] (0.2,1.3) --++(6,0);
\draw[line width = 2, color = red] (0.2,0.6) --++(6,0);

\draw[line width = 2, color = blue] (6.2,7.15) --++(-5.8,0)--++(0,-7.8)--++(-1.2,0);
\draw[line width = 2, color = blue] (0.4,-0.2) --++(5.3,0);

\draw[line width = 2, color = blue] (0.7,-0.2) --++(0,6.45)--++(0.2,0);
\draw[line width = 2, color = blue] (1.2,-0.2) --++(0,5.45)--++(0.2,0);
\draw[line width = 2, color = blue] (1.7,-0.2) --++(0,4.45)--++(0.2,0);

\draw[line width = 2, color = blue] (2.7,-0.2) --++(0,6.45)--++(0.2,0);
\draw[line width = 2, color = blue] (3.2,-0.2) --++(0,3.45)--++(0.2,0);
\draw[line width = 2, color = blue] (3.7,-0.2) --++(0,1.45)--++(0.2,0);

\draw[line width = 2, color = blue] (5.2,-0.2) --++(0,5.45)--++(0.2,0);
\draw[line width = 2, color = blue] (4.7,-0.2) --++(0,2.45)--++(0.2,0);
\draw[line width = 2, color = blue] (5.7,-0.2) --++(0,1.45)--++(0.2,0);

\end{tikzpicture}
\end{center}
\caption{Reference coloring is transported to every gadget. Circles denote the (in)equality gadgets, while rectangles denote the satisfiability gadgets.
Red and blue lines denote, respectively, pairs of overlapping segments with colors 1,2, and colors greater than 2. Segments $R_1,R_2,\ldots,R_k$ are positioned in the lower left corner of the picture.}
\label{fig:2dir-transport}
\end{figure}

Such a construction relies on a constant-size gadget, which allows us to turn or split the reference coloring. 
The construction of this gadget is depicted in Figure \ref{fig-turn}. 
Note that the number of segments in this gadget is constant if $k$ is constant. Moreover, turning or splitting the reference coloring of fewer than $k$ colors can be obtained by a simple adaptation of the turning gadget. 
Indeed, suppose we want to introduce a turning gadget for the set of colors $C \subseteq [k]$, with $|C|=k'<k$ and we have $k'$ overlapping segments carrying these colors. We introduce $k-k'$ dummy segments, overlapping these segments. The dummy segments will clearly receive colors from $[k] \setminus C$, but we do not know which segment will get which color.
Now we introduce a turning gadget for $k$ colors. We know which segments leaving the turning gadget get colors in $C$ (and we precisely know which segment gets which color). We do not need he remaining segments anymore, so we can finish them as soon as they leave the turning gadget.
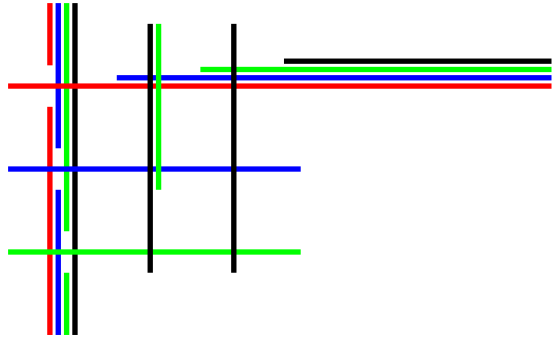
\begin{figure}[h]
\begin{center}
\begin{tikzpicture}[scale=1.1, line width=2]
\draw[color = red] (0,3) -- (0, 5.75);
\draw[color = red] (0,6.25) -- (0, 7);
\draw[color = blue] (0.1,3) -- (0.1, 4.75);
\draw[color = blue] (0.1,5.25) -- (0.1, 7);
\draw[color = green] (0.2,3) -- (0.2, 3.75);
\draw[color = green] (0.2,4.25) -- (0.2, 7);
\draw[color = black] (0.3,3) -- (0.3, 7);

\draw[color = red] (-0.5,6) -- (6,6);
\draw[color = blue] (-0.5,5) -- (3,5);
\draw[color = green] (-0.5,4) -- (3,4);

\draw[color = blue] (0.8,6.1) -- (6,6.1);
\draw[color = green] (1.8,6.2) -- (6,6.2);
\draw[color = black] (2.8,6.3) -- (6,6.3);

\draw[color = black] (1.2,3.75) -- (1.2,6.75);
\draw[color = green] (1.3,4.75) -- (1.3,6.75);

\draw[color = black] (2.2,3.75) -- (2.2,6.75);
\end{tikzpicture}
\end{center}
\caption{Turning gadget for $k=4$ colors. The parallel segments depicted close to each other are overlapping. Observe that the depicted 4-coloring is the only possible (up to the permutation of colors). For $k > 4$ the turning gadget is analogous.}
\label{fig-turn}
\end{figure}

Now, the only thing left is to connect every segment in every gadget to appropriate segments carrying the reference coloring (note that each $y^j_\ell$ belongs to some (in)equality gadget). This can easily be done using a constant number of additional segments per gadget: we introduce a turning gadget for $k$ colors, and finish the segments that are not needed anumore before they intersect the segments in gadgets (see Figure \ref{fig-lists}).

\begin{figure}[h]
\hfill
\begin{tikzpicture}[scale=0.8]
\draw[line width = 2, color = violet] (2,3) -- (5.2,3);
\draw[line width = 2, color = violet] (3.5,3) --++ (0,-1.2);
\draw[line width = 2, color = violet] (4,3) --++ (0,-0.7);
\draw[line width = 2, color = violet] (4.5,3) --++ (0,-0.7);
\draw (2,0) -- (6,0) node [above] {$x_i$};
\draw (5,-1) -- (5,4) node [left] {$y^j_\ell$};
\draw (6,2) -- (3,2) node [left] {$c$};
\draw (4,2.5)--(4,-0.5) node [below] {$a$};
\draw (4.5,2.5)--(4.5,-0.5) node [below] {$b$};
\end{tikzpicture}
\hfill
\begin{tikzpicture}[scale=0.7]
\draw[line width = 2, color = violet] (-4.5,4) -- (-0.5,4);
\draw[line width = 2, color = violet] (-3.8,4) --++ (0,-1.5);
\draw[line width = 2, color = violet] (-2.8,4) --++ (0,-1);
\draw[line width = 2, color = violet] (-1.8,4) --++ (0,-0.5);

\draw (0,3) -- (0,0) node [below] {$y^j_1$};
\draw (2,3.4) -- (2,0) node [below] {$y^j_2$};
\draw (4,3.8) -- (4,0) node [below] {$y^j_3$};

\draw (0.5,2.8) -- (-4,2.8) node [left] {$a$};
\draw (2.5,3.2) -- (-3,3.2) node [left] {$b$};
\draw (4.5,3.6) -- (-2,3.6) node [left] {$c$};

\draw (-1,4.5) -- (-1,2) node [below] {$d$};
\end{tikzpicture} 
\caption{Simulation of lists for vertices in (in)equality gadgets and satisfiability gadgets. Horizontal violet lines denote tuples of overlapping segments, carrying the reference coloring (all $k$ colors). In the vertical lines, segments carrying unnecessary colors are finished before they intersect the segments of the gadgets.}
\label{fig-lists}
\end{figure}
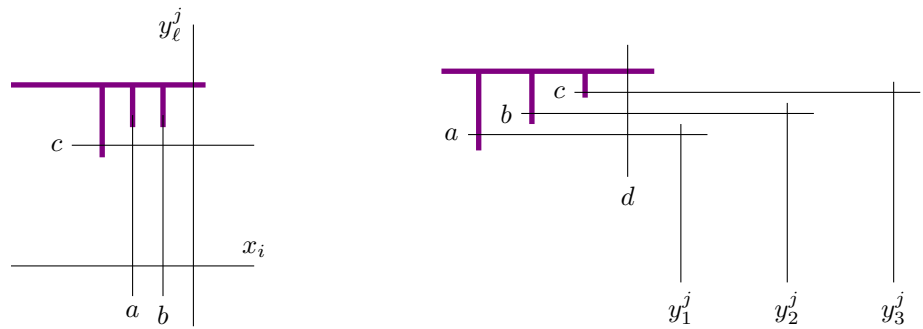

The total size of the construction increases by a constant factor, as we introduce $O(n)$ constant-size turning gadgets. 
Thus an algorithm for $k$-coloring the constructed \2dir graph in time $2^{o(n')}$ could be used to solve any \3sat instance in time $2^{o(n)}$, refuting the ETH.
\end{proof}

Observe that the construction in the proof of Theorem \ref{thm:2dir-lists} cannot be performed with segments of bounded lengths, since segments $x_i$ and $y^j_k$ need to have length $O(n)$ (while the segments inside the gadgets can have unit length).
For unit segments, we show the following weaker lower bound.

\begin{theorem} \label{thm:unit-2dir-lists}
For every $k\geq4$, \LkColoring of a unit \2dir graph or \kColoring of a unit \threedir
cannot be solved in time $2^{o(n^{2/3})}$, unless the ETH fails.
\end{theorem}

\begin{proof}
Consider a \3sat instance with variables $v_1,v_2,\ldots,v_n$ and clauses $C_1,C_2,\ldots,C_m$, where $m=\Theta(n)$. By duplicating some literals if necessary, we may assume that each clause contains exactly three literals.

Let us start with a reduction to \LkColoring in unit \2dir.
For each clause $C_j$ we introduce three vertical unit segments, each corresponding to one literal in $C_j$.
These segments will be called {\em literal segments}. We place them in such a way that the distance between the leftmost and the rightmost literal segment is slightly smaller than 1/2. 
The ordering of the segments is the following: first, the literal segments corresponding to the clause $C_1$, then literal segments corresponding to the clause $C_2$ and so on.
Moreover, they are slightly shifted vertically, so that the $y$-coordinates of their top endpoints form an increasing sequence. We set the list of possible colors for each literal segment to $\{3,4\}$ and for each three literal segments corresponding to a single clause, we introduce a satisfiability gadget already shown in Figure \ref{fig-sat}.
We will interpret 3 as assigning the value {\em true} to the particular literal (and 4 will correspond to {\em false}).
Figure \ref{fig-sat-placement} shows how the placement of literal segments and satisfiability gadgets.
\begin{figure}[h!]
\begin{tikzpicture}[scale=0.6]

\draw[red, line width = 1.5] (0,3) --++ (0,-2);
\draw[red, line width = 1.5] (2,3.3) --++ (0,-2.3);
\draw[red, line width = 1.5] (4,3.7) --++ (0,-2.7);
\draw[red] (0.5,2.8) --++ (-6,0);
\draw[red] (2.5,3.2) --++ (-8,0);
\draw[red] (4.5,3.6) --++ (-10,0);
\draw[red] (-1,3.7) --++ (0,-2.7);

\draw[blue, line width = 1.5] (6,4) --++ (0,-3);
\draw[blue, line width = 1.5] (8,4.3) --++ (0,-3.3);
\draw[blue, line width = 1.5] (10,4.7) --++ (0,-3.7);
\draw[blue] (6.5,3.8) --++ (-10,0);
\draw[blue] (8.5,4.2) --++ (-10,0);
\draw[blue] (10.5,4.6) --++ (-10,0);
\draw[blue] (5,4.7) --++ (0,-3.7);

\draw[violet, line width = 1.5] (12,5) --++ (0,-4);
\draw[violet, line width = 1.5] (14,5.3) --++ (0,-4.3);
\draw[violet, line width = 1.5] (16,5.7) --++ (0,-4.7);
\draw[violet] (12.5,4.8) --++ (-10,0);
\draw[violet] (14.5,5.2) --++ (-10,0);
\draw[violet] (16.5,5.6) --++ (-10,0);
\draw[violet] (11,5.7) --++ (0,-4.7);
\end{tikzpicture} 
\caption{Placement of satisfiability gadgets. Segments in one color correspond to one clause. Literal segments are depicted by thick lines.
Note that all segments can be freely extended to the left and to the bottom, to make them unit.
}
\label{fig-sat-placement}
\end{figure}
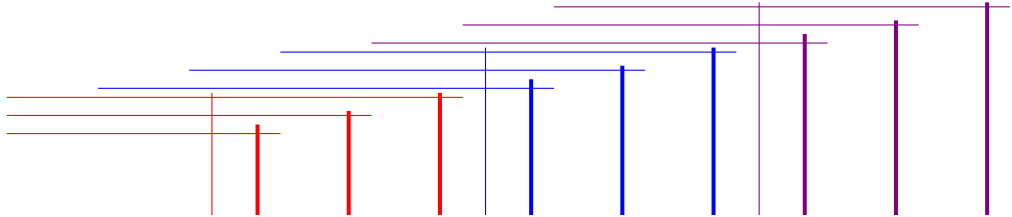
Analogously, for each variable $v$, we introduce a vertical segment for each occurrence of $v$ (we call these segments {\em occurrence segments}) with list $\{3,4\}$. The segments are placed in such a way that the distance between the leftmost and the rightmost occurrence segment is slightly smaller than 1/2. Moreover, leftmost occurence segments correspond to $v_1$, then we put the segments for $v_2$ etc.
For a variable $v_i$ for $i \in [n]$, we introduce a horizontal segment, intersecting the occurrence segments of variables $v_i,v_{i+1},\ldots,v_n$.
These segments will be called {\em variable segments}. The variable segments are pairwise disjoint and each has list $\{1,2\}$. Again, we will interpret 1 as the value {\em true} given to a variable, and 2 will denote {\em false}.
Now we need to ensure that the truth assignment defined by the coloring of occurrence segments is consistent.
For this, we will use a slightly modified version of the (in)equality segment introduced in Figure~\ref{fig-in-equality}. The modified gadget is shown in Figure~\ref{fig-in-equality-mod}.
\begin{figure}[h]
\begin{minipage}{.5\linewidth}%
\begin{tikzpicture}[scale=0.6]
\def\s{0.5}

\coordinate (yb) at (3,0) ;
\coordinate (ye) at (-5,0) ;
\coordinate (xb) at (0,-2.5) ;
\coordinate (xe) at (0,3.5) ;

\draw (yb) -- (ye) node [above] {$y$};
\draw (xb) -- (xe) node [left] {$x$};

\draw[dashed] (ye) --++(-\s,0) ;

\draw[dashed] (yb) --++(\s,0) ;
\draw[dashed] (xe) --++(0,\s) ;

\draw (-0.1,3) --++ (0,-6) node [below, left] {$a$};
\draw ( 0.1,3) --++ (0,-6) node [below, right] {$b$};
\end{tikzpicture} 
\end{minipage}%
\begin{minipage}{.3\linewidth}
{\small
\begin{tabular}{l | l}
vertex & list \\ \hline
$x$ & 3,4 \\
$y$ & 1,2 \\ \hline
$a$ & 3,2 \\
$b$ & 4,1 \\
\end{tabular}

\medskip
Lists in the\\equality gadget.
}
\end{minipage}%
\begin{minipage}{.3\linewidth}
{\small
\begin{tabular}{l | l}
vertex & list \\ \hline
$x$ & 3,4 \\
$y$ & 1,2 \\ \hline
$a$ & 3,1 \\
$b$ & 4,2 \\
\end{tabular}

\medskip
Lists in the\\inequality gadget.
}
\end{minipage}
\caption{Modified equality and inequality gadgets. The segment $x$ is an occurrence segment and the segment $y$ is a variable segment. Segments $a,b$, and $x$ overlap.}
\label{fig-in-equality-mod}
\end{figure}
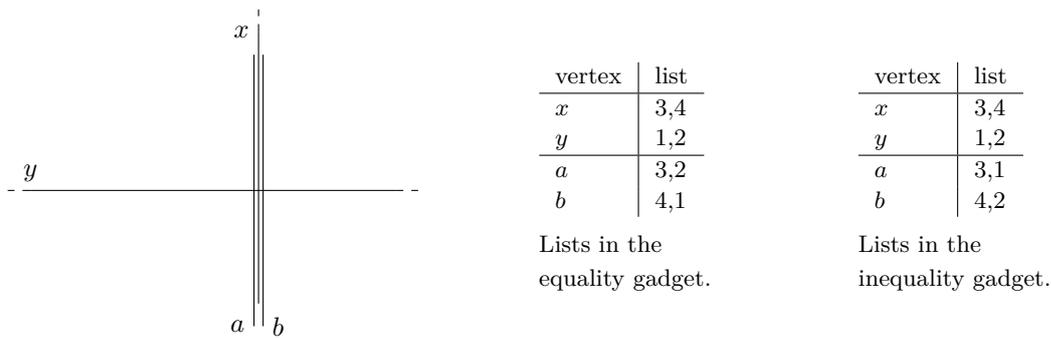
For a variable $v$ and its positive occurrence, we introduce an equality gadget joining the occurrence segment and the variable segment. Analogously, we introduce an inequality gadget joining the variable segment and the occurrence segment corresponding to a negative occurrence. The placement of all these segments is shown in Figure \ref{fig-equality-placement}.
\begin{figure}[h!]
\centering
\begin{tikzpicture}[scale=0.5]
\draw[red, line width = 1.5] (0,3) --++ (0,-6);
\draw[red, line width = 1.5] (1,3) --++ (0,-6);
\draw[red, line width = 1.5] (2,3) --++ (0,-6);
\draw[red] (-0.1,2.2) --++ (0,-5.5);
\draw[red] (0.1,2.2) --++ (0,-5.5);
\draw[red] (0.9,2.2) --++ (0,-5.5);
\draw[red] (1.1,2.2) --++ (0,-5.5);
\draw[red] (1.9,2.2) --++ (0,-5.5);
\draw[red] (2.1,2.2) --++ (0,-5.5);

\draw[blue, line width = 1.5] (3,3) --++ (0,-6);
\draw[blue, line width = 1.5] (4,3) --++ (0,-6);
\draw[blue] (2.9,1.7) --++ (0,-5);
\draw[blue] (3.1,1.7) --++ (0,-5);
\draw[blue] (3.9,1.7) --++ (0,-5);
\draw[blue] (4.1,1.7) --++ (0,-5);

\draw[violet, line width = 1.5] (5,3) --++ (0,-6);
\draw[violet, line width = 1.5] (6,3) --++ (0,-6);
\draw[violet, line width = 1.5] (7,3) --++ (0,-6);
\draw[violet] (4.9,1.2) --++ (0,-4.5);
\draw[violet] (5.1,1.2) --++ (0,-4.5);
\draw[violet] (5.9,1.2) --++ (0,-4.5);
\draw[violet] (6.1,1.2) --++ (0,-4.5);
\draw[violet] (6.9,1.2) --++ (0,-4.5);
\draw[violet] (7.1,1.2) --++ (0,-4.5);

\draw[red] (-0.5,2) --++ (10,0);
\draw[blue] (2.5,1.5) --++ (10,0);
\draw[violet] (4.5,1) --++ (10,0);

\end{tikzpicture} 
\caption{Placement of variable and occurrence segments along with (in)equality gadgets. Segments in one color correspond to one variable.
Literal segments are depicted by thick lines. Thin horizontal lines are variable segments and thin vertical lines are parts of (in)equality gadgets.}
\label{fig-equality-placement}
\end{figure}
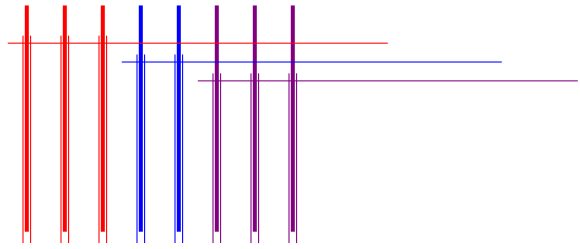

Now we need to make sure that the truth assignment given by the coloring of literal segments is consistent with the truth assignment given by the coloring of occurrence segments. We will do it in a very similar way as we synchronized the colorings of occurrence segments, i.e., by using auxiliary horizontal lines and equality gadgets.
Let $\ell_1,\ell_2,\ldots,\ell_{3m}$ denote the literals ordered as their corresponding literal segments (from left to right).
Let $o_1,o_2,\ldots,o_{3m}$ be the ordering of occurrences, again from left to right. Let $\sigma$ be the permutation of $[3m]$, such that the literal $\ell_i$ corresponds to the occurrence $o_{\sigma(i)}$.
Now, for every $i \in [n]$, we want to introduce an equality gadget between the literal segment corresponding to $\ell_i$ and the occurrence segment corresponding to $o_{\sigma(i)}$.

We observe that this is quite easy to do if $\sigma$ is either increasing or decreasing. As shown by Brandst\"{a}dt and Kratsch~\cite{BrandstadtK86}, each permutation of $[3m]$ can be partitioned into at most $z := \lceil \sqrt{6m+1/4}-1/2 \rceil = O(\sqrt{n})$ monotone sequences. Their proof is constructive and can be easily transformed into a polynomial algorithm finding such a partition. Let $\sigma_1,\sigma_2,\ldots,\sigma_z$ be the partition of $\sigma$, where each $\sigma_i$ is monotone.
We introduce $z$ {\em layers}, each corresponding to one $\sigma_i$. The $i$th layer is responsible to synchronize the colorings of literal segments and occurrence segments corresponding to elements of $\sigma_i$ (we call such segments {\em important} for layer $i$).
Figure~\ref{fig-layer} shows a single layer. Note that each layer contains copies of \emph{all} literal segments and occurrence segments. They appear in two groups -- literal segments on the left and occurrence segments on the right. The distance between leftmost and rightmost segment in one group is slightly less than 1/2, and the distance between the rightmost literal segment and the leftmost occurrence segment is slightly less than 2. This allows us to fit two equality gadgets, whose horizontal segments are collinear but non-overlapping (see Fig. \ref{fig-layer} and notice that we can adjust the distances within groups and between the groups, so that the distance between the horizontal segments is always positive and smaller than 1).
For each such a pair we introduce $k-1$ overlapping segments with lists $\{1,2,\ldots,k\}$, intersecting both of them and nothing else.
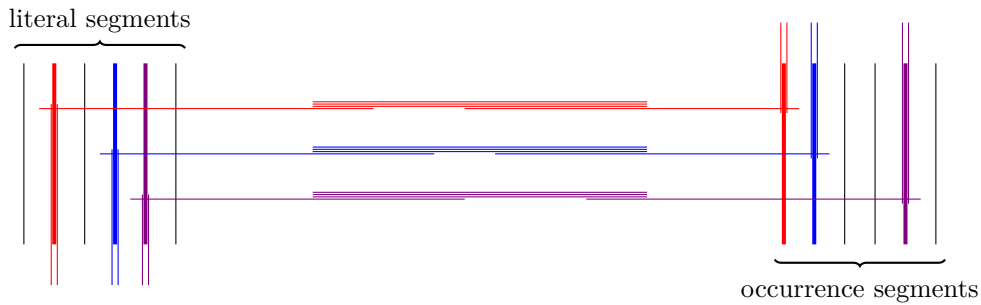
\begin{figure}[h!]
\centering
\begin{tikzpicture}[xscale=0.4,yscale=0.6]
\draw[] (0,3) --++ (0,-4);
\draw[red, line width = 1.5] (1,3) --++ (0,-4);
\draw[] (2,3) --++ (0,-4);
\draw[blue, line width = 1.5] (3,3) --++ (0,-4);
\draw[violet, line width = 1.5] (4,3) --++ (0,-4);
\draw[] (5,3) --++ (0,-4);

\draw[red, line width = 1.5] (25,3) --++ (0,-4);
\draw[blue, line width = 1.5] (26,3) --++ (0,-4);
\draw[] (27,3) --++ (0,-4);
\draw[] (28,3) --++ (0,-4);
\draw[violet, line width = 1.5] (29,3) --++ (0,-4);
\draw[] (30,3) --++ (0,-4);

\draw [thick,
    decoration={
        brace,        
        raise=0.2cm
    }, decorate] (-0.3,3) -- (5.3,3) node [pos=0.5,anchor=north,yshift=0.85cm] {literal segments}; 
    \draw [thick,
    decoration={
        brace,
        raise=0.2cm,
        mirror,
    }, decorate] (24.7,-1) -- (30.3,-1) node [pos=0.5,anchor=north,yshift=-0.35cm] {occurrence segments}; 

\draw[red] (0.5,2) --++ (11,0);    
\draw[red] (14.5,2) --++ (11,0); 

\draw[blue] (2.5,1) --++ (11,0); 
\draw[blue] (15.5,1) --++ (11,0); 

\draw[violet] (3.5,0) --++ (11,0);    
\draw[violet] (18.5,0) --++ (11,0);    

\draw[red] (0.9,2.1) --++ (0, -4);
\draw[red] (1.1,2.1) --++ (0, -4);
\draw[blue] (2.9,1.1) --++ (0, -3);
\draw[blue] (3.1,1.1) --++ (0, -3);
\draw[violet] (3.9,0.1) --++ (0, -2);
\draw[violet] (4.1,0.1) --++ (0, -2);

\draw[red] (24.9,1.9) --++ (0, 2);
\draw[red] (25.1,1.9) --++ (0, 2);
\draw[blue] (25.9,0.9) --++ (0, 3);
\draw[blue] (26.1,0.9) --++ (0, 3);
\draw[violet] (28.9,-0.1) --++ (0, 4);
\draw[violet] (29.1,-0.1) --++ (0, 4);

\draw[red] (9.5,2.05) --++ (11,0);    
\draw[red] (9.5,2.10) --++ (11,0);    
\draw[red] (9.5,2.15) --++ (11,0);    

\draw[blue] (9.5,1.05) --++ (11,0); 
\draw[blue] (9.5,1.10) --++ (11,0); 
\draw[blue] (9.5,1.15) --++ (11,0); 

\draw[violet] (9.5,0.05) --++ (11,0);    
\draw[violet] (9.5,0.10) --++ (11,0);    
\draw[violet] (9.5,0.15) --++ (11,0);    
\end{tikzpicture} 
\caption{A single layer $i$.
Thick colored segments indicate corresponding literal and occurrence segments, which are important for layer $i$.
Thin colored segments are parts of equality gadgets.}
\label{fig-layer}
\end{figure}

The last thing to do is to connect the literal segments (with (in)equality gadgets), layers, and occurrence segments with satisfiability gadgets. We place the occurrence segments at the bottom, and then we introduce layers in such a way that the corresponding vertical segments are collinear but non-intersecting. Finally, the literal segments with satisfiability gadgets are put on top.
Let $s_1$ and $s_2$ be collinear vertical segments and let $s_1$ be above $s_2$.
If  $s_1$ is an important literal segment of the $i$th layer, then there are two segments $a,b$ belonging to the appropriate equality gadget, which are covering the lower endpoint of $s_1$. We introduce $k-3$ vertical segments $q_1,q_2,\ldots,q_{k-3}$, intersecting only $s_1$,$a$,$b$ (so its upper endpoint is below the horizontal segment of the equality gadget, where $a$ and $b$ belong), and set their lists to $\{1,2,\ldots,k\}$. We adjust the distance between the layers, so that $a$, $b$, and $q$ intersect $s_2$, but no other vertex of its layer.

The situation is analogous if $s_2$ is an important occurrence segment of $i$th layer. Finally, if none of $s_1,s_2$ is an important segment of a layer, we introduce $k-1$ overlapping segments $q_1,q_2,\ldots,q_{k-1}$ with lists $\{1,2,\ldots,k\}$, intersecting $s_1$, $s_2$, and no other previously constructed segment. This way $s_1$ and $s_2$ are non-adjacent, but they are both intersected by three pairwise intersecting segments. This ensures that $s_1$ and $s_2$ will receive the same color.
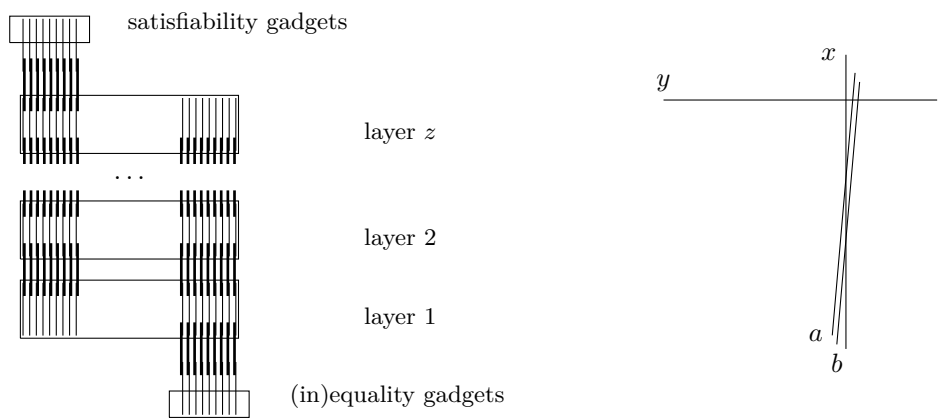
\begin{figure}[h!]
\centering
\begin{minipage}{.5\linewidth}%
\begin{tikzpicture}[scale=0.35]
\draw[] (6.0,0) --++ (0,2);
\draw[] (6.25,0) --++ (0,2);
\draw[] (6.5,0) --++ (0,2);
\draw[] (6.75,0) --++ (0,2);
\draw[] (7.0,0) --++ (0,2);
\draw[] (7.25,0) --++ (0,2);
\draw[] (7.5,0) --++ (0,2);
\draw[] (7.75,0) --++ (0,2);
\draw[] (8.0,0) --++ (0,2);

\draw (5.5,-0.1) --++ (3.0,0) --++ (0,1) --++ (-3, 0) --++ (0,-1) node [anchor=south,xshift=3cm] {\small (in)equality gadgets};

\draw[] (0.0,3) --++ (0,2);
\draw[] (0.25,3) --++ (0,2);
\draw[] (0.5,3) --++ (0,2);
\draw[] (0.75,3) --++ (0,2);
\draw[] (1.0,3) --++ (0,2);
\draw[] (1.25,3) --++ (0,2);
\draw[] (1.5,3) --++ (0,2);
\draw[] (1.75,3) --++ (0,2);
\draw[] (2.0,3) --++ (0,2);

\draw[] (6.0,3) --++ (0,2);
\draw[] (6.25,3) --++ (0,2);
\draw[] (6.5,3) --++ (0,2);
\draw[] (6.75,3) --++ (0,2);
\draw[] (7.0,3) --++ (0,2);
\draw[] (7.25,3) --++ (0,2);
\draw[] (7.5,3) --++ (0,2);
\draw[] (7.75,3) --++ (0,2);
\draw[] (8.0,3) --++ (0,2);

\draw (-0.1,2.9) --++ (8.2,0) --++ (0,2.2) --++ (-8.2, 0) --++ (0,-2.2) node [anchor=south,xshift=5cm] {\small layer 1};

\draw[] (0.0,6) --++ (0,2);
\draw[] (0.25,6) --++ (0,2);
\draw[] (0.5,6) --++ (0,2);
\draw[] (0.75,6) --++ (0,2);
\draw[] (1.0,6) --++ (0,2);
\draw[] (1.25,6) --++ (0,2);
\draw[] (1.5,6) --++ (0,2);
\draw[] (1.75,6) --++ (0,2);
\draw[] (2.0,6) --++ (0,2);

\draw[] (6.0,6) --++ (0,2);
\draw[] (6.25,6) --++ (0,2);
\draw[] (6.5,6) --++ (0,2);
\draw[] (6.75,6) --++ (0,2);
\draw[] (7.0,6) --++ (0,2);
\draw[] (7.25,6) --++ (0,2);
\draw[] (7.5,6) --++ (0,2);
\draw[] (7.75,6) --++ (0,2);
\draw[] (8.0,6) --++ (0,2);

\draw (-0.1,5.9) --++ (8.2,0) --++ (0,2.2) --++ (-8.2, 0) --++ (0,-2.2) node [anchor=south,xshift=5cm] {\small layer 2};

\draw[] (0.0,10) --++ (0,2);
\draw[] (0.25,10) --++ (0,2);
\draw[] (0.5,10) --++ (0,2);
\draw[] (0.75,10) --++ (0,2);
\draw[] (1.0,10) --++ (0,2);
\draw[] (1.25,10) --++ (0,2);
\draw[] (1.5,10) --++ (0,2);
\draw[] (1.75,10) --++ (0,2);
\draw[] (2.0,10) --++ (0,2);

\draw[] (6.0,10) --++ (0,2);
\draw[] (6.25,10) --++ (0,2);
\draw[] (6.5,10) --++ (0,2);
\draw[] (6.75,10) --++ (0,2);
\draw[] (7.0,10) --++ (0,2);
\draw[] (7.25,10) --++ (0,2);
\draw[] (7.5,10) --++ (0,2);
\draw[] (7.75,10) --++ (0,2);
\draw[] (8.0,10) --++ (0,2);

\draw (-0.1,9.9) --++ (8.2,0) --++ (0,2.2) --++ (-8.2, 0) --++ (0,-2.2) node [anchor=south,xshift=5cm] {\small layer $z$};

\draw[] (0.0,13) --++ (0,2);
\draw[] (0.25,13) --++ (0,2);
\draw[] (0.5,13) --++ (0,2);
\draw[] (0.75,13) --++ (0,2);
\draw[] (1.0,13) --++ (0,2);
\draw[] (1.25,13) --++ (0,2);
\draw[] (1.5,13) --++ (0,2);
\draw[] (1.75,13) --++ (0,2);
\draw[] (2.0,13) --++ (0,2);

\draw (-0.5,14.1) --++ (3.0,0) --++ (0,1) --++ (-3, 0) --++ (0,-1) node [anchor=south,xshift=3cm] {\small satisfiability gadgets};

\foreach \i in {0,...,8}
{
	\draw[line width=1] (\i/4+0.05,4.5) --++ (0,2);
	\draw[line width=1] (6+\i/4-0.05,4.5) --++ (0,2);
}
\foreach \i in {0,...,8}
{
	\draw[line width=1] (\i/4+0.05,11.5) --++ (0,2);
	\draw[line width=1] (6+\i/4-0.05,1.5) --++ (0,2);
}
\foreach \i in {0,...,8}
{
	\draw[line width=1] (\i/4+0.05,7.5) --++ (0,1);
	\draw[line width=1] (6+\i/4-0.05,7.5) --++ (0,1);
	\draw[line width=1] (\i/4+0.05,9.5) --++ (0,1);
	\draw[line width=1] (6+\i/4-0.05,9.5) --++ (0,1);
}
\node at (4,9) {$\ldots$};
\end{tikzpicture} 
\end{minipage}
\hfill
\begin{minipage}{.4\linewidth}%
\begin{tikzpicture}[scale=0.6]
\def\s{0.75}

\coordinate (yb) at (3,3) ;
\coordinate (ye) at (-3,3) ;
\coordinate (xb) at (1,-2.5) ;
\coordinate (xe) at (1,4) ;

\draw (yb) -- (ye) node [above] {$y$};
\draw (xb) -- (xe) node [left] {$x$};

\draw (1.2,3.6) --++ (-0.5,-5.8) node [left] {$a$};
\draw (1.3,3.4) --++ (-0.5,-5.8) node [below] {$b$};
\end{tikzpicture} 
\end{minipage}%
\caption{Overall construction (left). Modified equality gadgets in a unit 3-DIR graph (right).}
\label{fig-unit-overall}
\end{figure}
The number of segments in each layer is $O(n)$ and the number of layers is $z=O(\sqrt{n})$. Thus the total number of segments in our construction is $O(n^{3/2})$. This implies that an algorithm solving \LkColoring on unit \2dir graphs with $N$ vertices in time $2^{o(N^{2/3})}$ could be used to solve \3sat with $n$ variables in time $2^{o(n)}$, which contradicts the ETH.

If we want to obtain a reduction to the non-list \kColoring for any $k \geq 4$, we need to transport the reference coloring to each gadget.
However, this cannot be done for segments $a,b$ in our (in)equality gadgets (recall Theorem \ref{thm:2dir} and Figure \ref{fig-in-equality-mod}), which have non-trivial lists and are fully covered by other segments -- note that if a segment in color 4 intersects $a$, it will also intersect $x$ (note that the segments with lists $\{1,2,\ldots,k\}$, even if they are fully covered, are not problematic since they do not need the reference coloring).
But if we use a third direction, we can make $a,b$ intersect $x$ and $y$ (again, we use notation in Figure \ref{fig-in-equality-mod}) and no other vertex with non-trivial list in our grid-like structure -- it is enough to choose their slope to be ``almost vertical'' (see Fig.~\ref{fig-unit-overall}).  This shows the claimed lower bound, for every $k \geqslant 4$, for \Coloring{k} of unit \threedir graphs and completes the proof.
\end{proof}

We show that on segment graphs, \mds, \mcds, and \mids are unlikely to admit a subexponential algorithm.

\begin{theorem}\label{thm:mds-seg}
\textsc{Min (Connected) Dominating Set} cannot be solved in time $2^{o(n)}$ on segment graphs with $n$ vertices, unless the ETH fails.
\end{theorem}

\begin{proof}
Using standard tricks we can transform an arbitrary \3sat formula $\psi$ with $N'$ variables and $M'$ clauses
into an equivalent {\sc Cnf-Sat} formula $\phi$ with $N=O(N')$ variables and $M=O(M')$ clauses, where each variable appears exactly twice positively and twice negatively.

Put $M$ pairwise-disjoint small segments on a circle as shown in Figure~\ref{fig:mds-seg}; $M$ slightly perturbed points work as long as no two pairs define the same direction. Each small segment $s(C_j)$ represents a distinct clause $C_j$.
For each literal $\sigma x_i$, where $x_i$ is one of the $N$ variables appearing in $\phi$ and $\sigma \in \{\emptyset,\neg\}$, we add a segment $s(\sigma x_i)$ crossing only the two small segments $s(C_j)$ and $s(C_{j'})$ corresponding to the clauses this literal satisfies.
We prolong all the segments corresponding to literals to make them pairwise intersect.

For each pair of literals $x_i, \neg x_i$, we add a small segment $s(i)$ near the intersection of $s(x_i)$ and $s(\neg x_i)$ which intersects only $s(x_i)$ and $s(\neg x_i)$. This finishes the construction (see Figure~\ref{fig:mds-seg}).
Note that the total number of segments is $n := 3N+M=\Theta(N)$.
We claim that there is a dominating set of size $N$ in the intersection graph if and only if $\phi$ is satisfiable.

Indeed, to dominate all the segments $s(i)$ for $i \in [N]$, one has to take exactly one of $s(x_i)$ and $s(\neg x_i)$.
The choice of the literals will only dominate all the segments $s(C_j)$ for $j \in [M]$ if the chosen dominating set corresponds to a satisfying assignment. The reverse direction is straightforward.

Moreover, as every pair of segments representing literals intersects, the dominating set encoding the satisfiable assignment is connected (it even induces a clique).
   \begin{figure}[h!]
    \centering
    \begin{tikzpicture}
      \def\m{9}
      \def\n{6}
      \foreach \i in {1,...,\m}{
        \begin{scope}
          \node at (360/\m * \i - 360/\m:3) {$s(C_{\i})$} ;
          \draw (360/\m * \i - 360/\m+5:2.5) -- (360/\m * \i - 360/\m-5:2.5) ;
        \end{scope}
      }
      \foreach \i/\j/\k in {1/4/black,2/5/red,3/7/blue,4/9/green,1/8/green,2/7/purple,3/5/purple,4/6/red,1/6/blue,2/8/black,3/9/cyan,5/7/cyan}{
        \draw[color=\k] (360/\m * \i - 360/\m:2.55) -- (360/\m * \j - 360/\m:2.55) ;
      }
      \draw (-0.8,-0.5) -- (-0.6,-0.7) ;
      \node at (-1,-0.4) {\footnotesize{$\sigma(1)$}} ;
      \begin{scope}[xshift=-1cm,yshift=1.6cm]
        \draw (-0.8,-0.5) -- (-0.6,-0.7) ;
        \node at (-0.5,-0.85) {\footnotesize{$\sigma(2)$}} ;
      \end{scope}
      \draw (1.4,-1.22) -- (1.7,-1.22) ;
      \node at (1.55,-1.45) {\footnotesize{$\sigma(3)$}} ;
      \draw (1.43,0.52) -- (1.73,0.52) ;
      \node at (1.58,0.3) {\footnotesize{$\sigma(4)$}} ;

      \node at (-0.7,0.4) {\footnotesize{$\sigma(x_1)$}} ;
      \node at (0.1,-0.6) {\footnotesize{$\sigma(\neg x_1)$}} ;
    \end{tikzpicture}
    \caption{An example with 9 clauses and 6 variables. The segments $\sigma(x_i)$ and $\sigma(\neg x_i)$ can be inferred from $\sigma(i)$ or from the colors (although we only specified which is which for variable $x_1$). Observe that the light blue and purple pairs do not intersect inside the circle so those segments should be prolonged outside it until they meet (this part is not shown in the picture).}
    \label{fig:mds-seg}
  \end{figure}
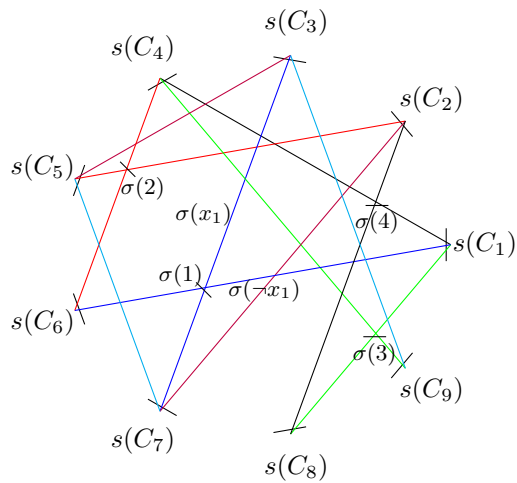  
\end{proof}

The hardness proof for \mids has to be quite different and is trickier.

\begin{theorem} \label{thm:mids-seg}
\mids cannot be solved in time $2^{o(n)}$ on segment graphs with $n$ vertices, unless the ETH fails.
\end{theorem}

\begin{proof}
  We reduce from \textsc{3-SAT} on $N$ variables and $M$ clauses, where each literal appears at most four times, and each clause contains \emph{exactly} three literals.
  This restriction was shown NP-complete by Tovey \cite{Tovey84}.
  We can further assume that each \emph{literal} appears at most three times.
  Indeed, if the same literal appears four times, then by definition its negation cannot appear in the whole instance.
  Hence, the variable can be set so as to satisfy the four corresponding clauses.
  
The variable gadget $G(x_i)$ for variable $x_i$ consists of three parallel segments $T_i$ representing positive occurrences crossing three parallel segments $F_i$ representing negative occurrences.
Those two sets of segments intersect three dummy pairs of parallel segments as shown on Figure~\ref{fig:mids-seg-gadgets}.
Note that even if a literal appears strictly fewer than three times, we keep exactly three parallel segments to encode it.

The clause gadget $G(C_j)$ for the clause $C_j$ consists of three pairwise crossing segments, drawn in blue in Figure~\ref{fig:mids-seg-gadgets}, each of which crosses one red segment. Each pair of blue and red segment corresponds to one of literals of $C_j$.
Additionally, all blue segments are intersected by four parallel dummy segments, and each pair of crossing blue and red segment has a private segment, crossing both of them.

  \begin{figure}[h!]
    \centering
    \begin{tikzpicture}[rotate=90, scale = 1.5]
        \draw (-0.4,0.4) -- (2,-0.1) ;
        \draw (-0.2,0.25) -- (2,-0.2) ;
        \draw (0,0.1) -- (2,-0.3) ;
        \draw (-0.4,-0.4) -- (2,0.1) ;
        \draw (-0.2,-0.25) -- (2,0.2) ;
        \draw (0,-0.1) -- (2,0.3) ;

        \draw (0.05,-0.15) -- (0.05,0.15) ;
        \draw (0.12,-0.15) -- (0.12,0.15) ;
      
        \draw (-0.08,-0.3) -- (-0.08,0.3) ;
        \draw (-0.15,-0.3) -- (-0.15,0.3) ;
        
        \draw (-0.23,-0.45) -- (-0.23,0.45) ;
        \draw (-0.3,-0.45) -- (-0.3,0.45) ;
        
        \node at (1.8,0.5) {$F_i$} ;
        \node at (1.8,-0.5) {$T_i$} ;
     
      \end{tikzpicture}
      \hskip 3cm
      \begin{tikzpicture}[scale = 2]
		\draw (-0.5,0.2) --++ (1,0);      
		\draw (-0.5,0.25) --++ (1,0);
		\draw (-0.5,0.3) --++ (1,0);
		\draw (-0.5,0.35) --++ (1,0);
      
		\draw[color=blue] (-0.5,0.4) --++ (1,-1);   
		\draw[color=blue] (0,0.4) --++ (0,-1);  
		\draw[color=blue] (0.5,0.4) --++ (-1,-1);
		
		\draw[color=red] (-0.5,-0.3) --++ (0.3,-1);  
		\draw[color=red] (-0.05,-0.3) --++ (0.3,-1);
		\draw[color=red] (0.3,-0.3) --++ (0.3,-1);
		
		\draw (-0.5,-0.5) --++ (0.2,0);
		\draw (-0.1,-0.5) --++ (0.2,0);
		\draw (0.3,-0.5) --++ (0.2,0);		      

      
    \end{tikzpicture}
    \caption{The variable gadget $G(x_i)$ for the variable $x_i$ (left) and the clause gadget $G(C_j)$ for the clause $C_j$ (right).}
    \label{fig:mids-seg-gadgets}
  \end{figure}
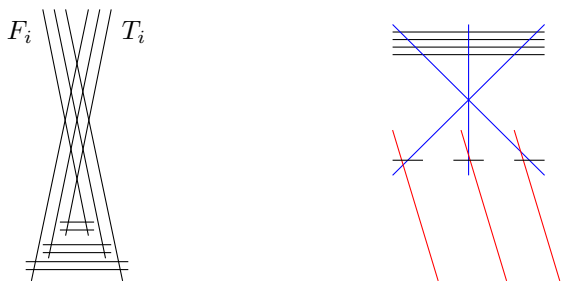

Now, for every literal $(\neg)x_i$ belonging to the clause $C_j$, we add a new segment denoted by $(\neg)x_i \in C_j$.
This segment crosses one non-dummy segment in $G(x_i)$: a segment of $T_i$ if the literal is positive, and a segment of $F_i$ otherwise. Moreover, it crosses one literal (red) segment from $G(C_j)$, and no other segments in variable and clause gadgets, see Figure~\ref{fig:mids-seg-overall}. 
Those lastly introduced segments cross exactly once each literal segment and at most once each segment of $\bigcup_{i \in [n]} T_i \cup F_i$.

We claim that such a constructed graph has an independent dominating set of size at most $3N+3M$ if and only if the initial formula is satisfiable.
  
  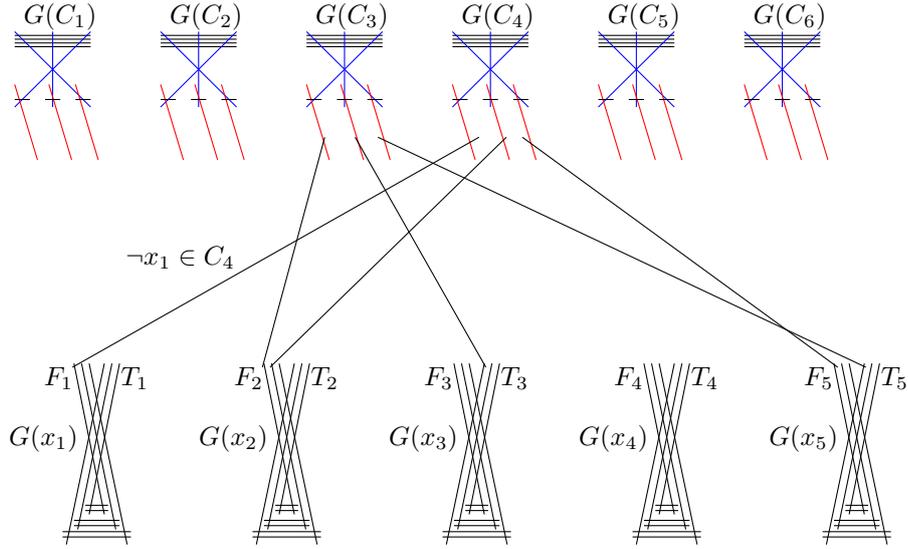
\begin{figure}[h!]
    \centering
    \begin{tikzpicture}
      \def\n{5}
      \def\m{6}
      \foreach \i in {1,...,\n}{ 
      \begin{scope}[xshift=2.5 * \i cm, rotate=90]    
      
        \draw (-0.4,0.4) -- (2,-0.1) ;
        \draw (-0.2,0.25) -- (2,-0.2) ;
        \draw (0,0.1) -- (2,-0.3) ;
        \draw (-0.4,-0.4) -- (2,0.1) ;
        \draw (-0.2,-0.25) -- (2,0.2) ;
        \draw (0,-0.1) -- (2,0.3) ;

        \draw (0.05,-0.15) -- (0.05,0.15) ;
        \draw (0.12,-0.15) -- (0.12,0.15) ;
      
        \draw (-0.08,-0.3) -- (-0.08,0.3) ;
        \draw (-0.15,-0.3) -- (-0.15,0.3) ;
        
        \draw (-0.23,-0.45) -- (-0.23,0.45) ;
        \draw (-0.3,-0.45) -- (-0.3,0.45) ;
        
        \node at (1.8,0.5) {$F_\i$} ;
        \node at (1.8,-0.5) {$T_\i$} ;

        \node at (1,0.7) {$G(x_\i)$} ;
      \end{scope}
      }

      \foreach \i in {1,...,\m}{ 
      \begin{scope}[xshift=1.92 * \i cm, yshift=6cm]
		\draw (-0.5,0.2) --++ (1,0);      
		\draw (-0.5,0.25) --++ (1,0);
		\draw (-0.5,0.3) --++ (1,0);
		\draw (-0.5,0.35) --++ (1,0);
      
		\draw[color=blue] (-0.5,0.4) --++ (1,-1);   
		\draw[color=blue] (0,0.4) --++ (0,-1);  
		\draw[color=blue] (0.5,0.4) --++ (-1,-1);
		
		\draw[color=red] (-0.5,-0.3) --++ (0.3,-1);  
		\draw[color=red] (-0.05,-0.3) --++ (0.3,-1);
		\draw[color=red] (0.3,-0.3) --++ (0.3,-1);
		
		\draw (-0.5,-0.5) --++ (0.2,0);
		\draw (-0.1,-0.5) --++ (0.2,0);
		\draw (0.3,-0.5) --++ (0.2,0);	
      \node at (0.1,0.6) {$G(C_\i)$} ;

      \end{scope}
      }

      \def\a{1.95}
      \def\b{5}
      \draw (4.68,\a) -- (5.5,\b) ;
      \draw (7.62,\a) -- (5.9,\b) ;
      \draw (12.62,\a) -- (6.2,\b) ;

      \draw (2.18,\a) -- (7.53,\b) ;
      \draw (4.78,\a) -- (7.89,\b) ;
      \draw (12.25,\a) -- (8.1,\b) ;

      \node at (3.6,3.4) {$\neg x_1 \in C_4$} ;
    \end{tikzpicture}
    \caption{The overall picture. We only represented two clauses: $C_3=\neg x_2 \lor x_3 \lor x_5$ and $C_4=\neg x_1 \lor \neg x_2 \lor \neg x_5$.}
    \label{fig:mids-seg-overall}
  \end{figure}

%

First, suppose that $A$ is a satisfying assignment. If $x_i$ is set to true by~$A$, we select the three segments of $T_i$ in the solution, otherwise we select the three segments of $F_i$. 
Now consider a clause $C_j$.
Since $A$ is satisfying, it contains at least one true literal.
In $G(C_j)$ we select the blue segment corresponding to a true literal and the red segments corresponding to the other two literals.
Note that this way we select $3N+3M$ segments and the selected set dominates all segments from all gadgets.
Let us consider a segment $(\neg)x_i \in C_j$.
Either it is dominated by one of the selected red segments in $G(C_j)$.
Or it corresponds to a true literal of $C_j$, so it is dominated by a selected segment in $G(x_i)$.

On the other hand, assume that there is an independent dominating set $S$ of size at most $3N + 3M$.
Notice that in order to dominate all dummy segments in $G(x_i)$, we need to select at least three segments from $G(x_i)$, and if we want to select exactly three, we need to choose either all segments in $T_i$, or all segments in $F_i$.
Analogously, we need to select at least three segments from each $G(C_j)$, and if we want to select exactly three, we need do choose one of blue segments and thus we cannot choose its corresponding red segment.
Note that since the total size of $S$ is $3N+3M$, we need to select exactly three segments in each gadget, and no segment $(\neg)x_i \in C_j$ is selected.

We define the assignment $A$ as follows: if segments $T_i$ are in $S$, then $x_i$ is set true, and otherwise $x_i$ is set false. Suppose that $C_j$ is not satisfied by $A$, which means that all its literals are false. This means that the three segments joining appropriate variable gadgets with $G(C_j)$ are not dominated by segments in variable gadgets, so $S$ must contain all three red segments from $G(C_j)$. However, this way the horizontal segments from $G(C_j)$ are not dominated, a contradiction.

The total number of segments is bounded by $12N+13M+3M=O(N+M)$, so the claim holds.
\end{proof}

\begin{theorem} \label{thm:mcli-str} 
\mcli cannot be solved in time $2^{o(n)}$ on strings graphs with $n$ vertices, unless the ETH fails.
\end{theorem}

\begin{proof}
We reduce from \3sat with a linear number of clauses, where every clause contains exactly three literals.
Let $\phi$ be an instance with $N$ variables and $M=\Theta(N)$ clauses.
For any positive integers $p$ and $s$, the {\em co-cluster} $\text{CC}_{p,s}=K_{s,s,\ldots,s~\text{(p times)}}$ can be represented as in Figure \ref{fig:mcli-str-cocluster}.
\begin{figure}[h!]
    \centering
    \begin{tikzpicture}
      \def\hb{-3}
      \def\he{3}
      \foreach \i in {-3,...,3}{
        \begin{scope}[rotate=10*\i]
          \draw (\hb,0) -- (\he,0) ;
          \draw (\hb,0.1) -- (\he,0.1) ;
          \draw (\hb,0.2) -- (\he,0.2) ;
        \end{scope}
      }
    \end{tikzpicture}
    \caption{Realization of a co-cluster $\text{CC}_{p,s}=K_{s,s,...,s}$ with $s=3$ and $p=7$.}
    \label{fig:mcli-str-cocluster}
\end{figure}
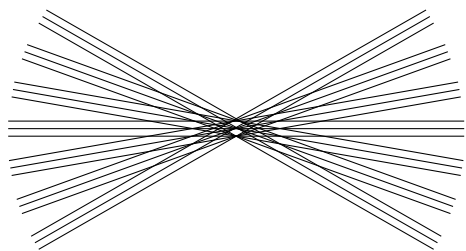
We encode the $N$ variables by $2N$ curves representing \emph{true} and \emph{false} for each variable by a co-cluster $\text{CC}_{N,2}$ and the $M=\Theta(N)$ clauses by $3M$ curves each representing a distinct literal in a clause by a co-cluster $\text{CC}_{M,3}$.  There are in total $n := 2N+3M=\Theta(N)$ curves.
We make the $2N$ \emph{variable} curves intersect the $3M$ \emph{literal} curves in a grid-like way.
They form an almost complete biclique $K_{2N,3M}$ where $3M$ edges are removed.

More precisely, the \emph{literal} curve $c(l^j_i)$ ($j \in [M]$, $i \in [3]$) intersects every \emph{variable} curve but $c(\sigma x_k)$ ($k \in [N]$) encoding the $k$-th variable with sign $\sigma \in \{\emptyset,\neg\}$ for which $l^j_i$ and $\sigma w_k$ are opposite literals (see Figure~\ref{fig:mcli-str-overall}).
\begin{figure}[h!]
    \centering
    \begin{tikzpicture}[scale=0.5]
      \def\v{6}
      \def\hb{0}
      \def\he{10}
      \def\vb{-0.5}
      \def\ve{8}   
      \def\o{0.45}
      
      \foreach \i in {1,...,\v}{
        \node at (\hb-0.9,7-\i) {\scriptsize{$c(\neg x_\i)$}} ;
        \node at (\hb-0.7,7-\i+\o) {\scriptsize{$c(x_\i)$}} ;
        \begin{scope}[rotate=0.2*\i-0.6]
        \draw (\hb,\i) -- (\he,\i) ;
        \draw (\hb,\i+\o) -- (\he,\i+\o) ;
        \end{scope}
      }

      \foreach \i/\j/\k/\h in {1/4/-2/0,1.25/6/-1.5/0,1.5/1.5/-1/0, 3/2/-0.3/2,3.25/3.5/0.2/2,3.5/5.5/0.7/2, 5/4.5/1.5/4,5.25/3/2/4,5.5/1.5/2.5/4}{
       \draw (\i,\j+0.25) -- (\i,\ve-0.5*\k) -- (\he+2-0.5*\k+\h,\ve-0.5*\k) -- (\he+2-0.5*\k+\h,\vb+0.5*\k) -- (\i,\vb+0.5*\k) -- (\i,\j-0.25) ;
      }
      \node at (3,9.3) {\scriptsize{$x_3 \lor x_1 \lor \neg x_6$}} ;
    \end{tikzpicture}
    \caption{The representation of 3 clauses: $x_3 \lor x_1 \lor \neg x_6$, $x_5 \lor \neg x_4 \lor \neg x_2$, and $\neg x_3 \lor x_4 \lor \neg x_6$.}
    \label{fig:mcli-str-overall}
  \end{figure}
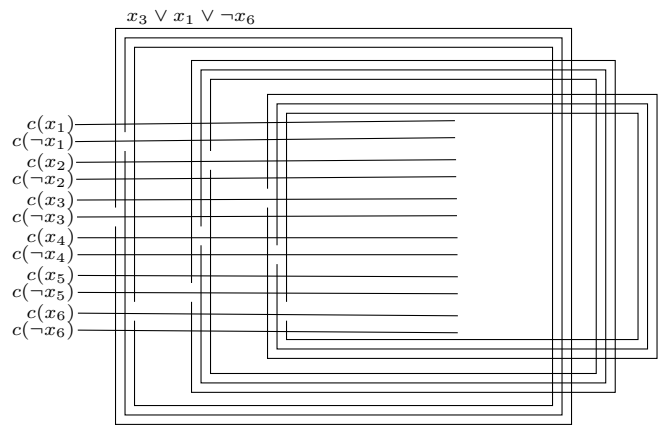
It is easy to observe that there is a clique of size $N+M$ if and only if $\phi$ has a satisfying assignment.
  \end{proof}


\section{Perspectives}
We have started a generalized optimality program on segment and string graphs for the most principal graph problems.
On the algorithmic side, we extended a subexponential algorithm for \mis on string graphs \cite{FoxP11} to two other problems: \Coloring{3} and \fvs.
On the complexity side, we showed that no subexponential algorithm is likely for, among others, \Coloring{4}, \mds, and \mids. 
It is quite easy to obtain such lower bounds for string graphs.
Extending those results to segments requires more ingenuity, and even more so when it comes to unit segments. 

A handful of questions remains unsettled.
Can we improve the algorithm or give tight ETH lower bounds for the following problems: \mis without geometric representation, \Coloring{3}, and \fvs on segments/strings?
Can we show for \mcli the same lower bound for segment graphs as we have for string graphs.
The mere NP-hardness of \mcli on segments \cite{CabelloCL13} answered a long-standing open question.
Hence, it is likely that getting a tight ETH hardness will be difficult.
We would also find interesting to have, for a certain problem, an algorithm for segments (resp. unit segments) which beats the ETH lower bound on strings (resp. segments).
So far, we only have candidate problems for such a \emph{``separation''}.

Finally, another natural continuation of this work is to determine which fixed-parameter tractable problems can be solved in time $O^*(2^{\tilde{O}(k^{2/3})})$ or $O^*(2^{\tilde{O}(\sqrt k)})$, and which W[1]-hard problems can be solved in time $f(k)n^{O(\sqrt k)}$ on segments and strings.
For instance, \mvc can be solved in time $O^*(2^{\tilde{O}(k^{2/3})})$ (even in time $O^*(2^{\tilde{O}(\sqrt k)})$ if a geometric representation is given with $O^*(2^{\tilde{O}(\sqrt k)})$ intersections) on string graphs due to the linear kernel yielding an equivalent instance on $2k$ vertices and the algorithm for \mis.
The latter problem can be solved in $n^{O(\sqrt k)}$ in segments or more generally in polygons of polynomial complexity \cite{MarxP15}, while \mds on string graphs cannot be solved in time $f(k)n^{o(k)}$, for any computable function $f$, unless the ETH fails (since this lower bound holds for split graphs).

\bibliographystyle{abbrv}

\begin{thebibliography}{10}

\bibitem{AdamaszekHW17}
A.~Adamaszek, S.~Har{-}Peled, and A.~Wiese.
\newblock Approximation schemes for independent set and sparse subsets of
  polygons.
\newblock {\em CoRR}, abs/1703.04758, 2017.

\bibitem{AdamaszekW14}
A.~Adamaszek and A.~Wiese.
\newblock A {QPTAS} for maximum weight independent set of polygons with
  polylogarithmically many vertices.
\newblock In {\em Proceedings of the Twenty-Fifth Annual {ACM-SIAM} Symposium
  on Discrete Algorithms, {SODA} 2014, Portland, Oregon, USA, January 5-7,
  2014}, pages 645--656, 2014.

\bibitem{AlberF04}
J.~Alber and J.~Fiala.
\newblock Geometric separation and exact solutions for the parameterized
  independent set problem on disk graphs.
\newblock {\em J. Algorithms}, 52(2):134--151, 2004.

\bibitem{Baker}
B.~S. Baker.
\newblock Approximation algorithms for np-complete problems on planar graphs.
\newblock {\em J. ACM}, 41(1):153--180, Jan. 1994.

\bibitem{Biro17}
C.~Bir{\'{o}}, {\'{E}}.~Bonnet, D.~Marx, T.~Miltzow, and
  P.~Rz{\k{a}}{\.{z}}ewski.
\newblock Fine-grained complexity of coloring unit disks and balls.
\newblock In {\em 33rd International Symposium on Computational Geometry, SoCG
  2017, July 4-7, 2017, Brisbane, Australia}, pages 18:1--18:16, 2017.

\bibitem{DBLP:journals/jgaa/BoyerM04}
J.~M. Boyer and W.~J. Myrvold.
\newblock On the cutting edge: Simplified o(n) planarity by edge addition.
\newblock {\em J. Graph Algorithms Appl.}, 8(2):241--273, 2004.

\bibitem{BrandstadtK86}
A.~Brandst{\"{a}}dt and D.~Kratsch.
\newblock On partitions of permutations into increasing and decreasing
  subsequences.
\newblock {\em Elektronische Informationsverarbeitung und Kybernetik},
  22(5/6):263--273, 1986.

\bibitem{CabelloCL13}
S.~Cabello, J.~Cardinal, and S.~Langerman.
\newblock The clique problem in ray intersection graphs.
\newblock {\em Discrete {\&} Computational Geometry}, 50(3):771--783, 2013.

\bibitem{ChalopinG09}
J.~Chalopin and D.~Gon{\c{c}}alves.
\newblock Every planar graph is the intersection graph of segments in the
  plane: extended abstract.
\newblock In {\em Proceedings of the 41st Annual {ACM} Symposium on Theory of
  Computing, {STOC} 2009, Bethesda, MD, USA, May 31 - June 2, 2009}, pages
  631--638, 2009.

\bibitem{DemaineFHT04}
E.~D. Demaine, F.~V. Fomin, M.~T. Hajiaghayi, and D.~M. Thilikos.
\newblock Bidimensional parameters and local treewidth.
\newblock {\em SIAM J. Discrete Math.}, 18(3):501--511, 2004.

\bibitem{DemaineFHT05b}
E.~D. Demaine, F.~V. Fomin, M.~T. Hajiaghayi, and D.~M. Thilikos.
\newblock Subexponential parameterized algorithms on bounded-genus graphs and
  ${H}$-minor-free graphs.
\newblock {\em J. ACM}, 52(6):866--893, 2005.

\bibitem{DemaineH08}
E.~D. Demaine and M.~Hajiaghayi.
\newblock The bidimensionality theory and its algorithmic applications.
\newblock {\em Comput. J.}, 51(3):292--302, 2008.

\bibitem{DemaineH08a}
E.~D. Demaine and M.~Hajiaghayi.
\newblock Linearity of grid minors in treewidth with applications through
  bidimensionality.
\newblock {\em Combinatorica}, 28(1):19--36, 2008.

\bibitem{DemaineH04}
E.~D. Demaine and M.~T. Hajiaghayi.
\newblock Fast algorithms for hard graph problems: Bidimensionality, minors,
  and local treewidth.
\newblock In {\em GD 2014 Proc.}, pages 517--533, 2004.

\bibitem{Erdos1987}
P.~Erd\H{o}s and G.~Szekeres.
\newblock {\em A Combinatorial Problem in Geometry}, pages 49--56.
\newblock Birkh{\"a}user Boston, Boston, MA, 1987.

\bibitem{FominLPSZ17}
F.~V. Fomin, D.~Lokshtanov, F.~Panolan, S.~Saurabh, and M.~Zehavi.
\newblock Finding, hitting and packing cycles in subexponential time on unit
  disk graphs.
\newblock {\em CoRR}, abs/1704.07279, 2017.

\bibitem{FominLS12}
F.~V. Fomin, D.~Lokshtanov, and S.~Saurabh.
\newblock Bidimensionality and geometric graphs.
\newblock In {\em SODA 2012 Proc.}, pages 1563--1575, 2012.

\bibitem{FoxP11}
J.~Fox and J.~Pach.
\newblock Computing the independence number of intersection graphs.
\newblock In {\em Proceedings of the Twenty-Second Annual {ACM-SIAM} Symposium
  on Discrete Algorithms, {SODA} 2011, San Francisco, California, USA, January
  23-25, 2011}, pages 1161--1165, 2011.

\bibitem{FoxPachArxiv}
J.~Fox and J.~Pach.
\newblock Applications of a new separator theorem for string graphs.
\newblock {\em CoRR}, abs/1302.7228, 2013.

\bibitem{DBLP:journals/jct/FoxPT10}
J.~Fox, J.~Pach, and C.~D. T{\'{o}}th.
\newblock A bipartite strengthening of the crossing lemma.
\newblock {\em J. Comb. Theory, Ser. {B}}, 100(1):23--35, 2010.

\bibitem{HarPeled14}
S.~Har{-}Peled.
\newblock Quasi-polynomial time approximation scheme for sparse subsets of
  polygons.
\newblock In {\em 30th Annual Symposium on Computational Geometry, SOCG'14,
  Kyoto, Japan, June 08 - 11, 2014}, page 120, 2014.

\bibitem{ImpagliazzoPaturi}
R.~Impagliazzo and R.~Paturi.
\newblock On the complexity of $k$-{SAT}.
\newblock {\em J. Comput. Syst. Sci.}, 62(2):367 -- 375, 2001.

\bibitem{Sparsification}
R.~Impagliazzo, R.~Paturi, and F.~Zane.
\newblock Which problems have strongly exponential complexity?
\newblock In {\em FOCS 1998 Proc.}, pages 653--662, Nov 1998.

\bibitem{DBLP:journals/jcss/ImpagliazzoPZ01}
R.~Impagliazzo, R.~Paturi, and F.~Zane.
\newblock Which problems have strongly exponential complexity?
\newblock {\em J. Comput. Syst. Sci.}, 63(4):512--530, 2001.

\bibitem{DBLP:journals/jct/Kratochvil91a}
J.~Kratochv{\'{\i}}l.
\newblock String graphs. {II.} recognizing string graphs is np-hard.
\newblock {\em J. Comb. Theory, Ser. {B}}, 52(1):67--78, 1991.

\bibitem{KratochvilM91}
J.~Kratochv{\'{\i}}l and J.~Matou{\v{s}}ek.
\newblock String graphs requiring exponential representations.
\newblock {\em J. Comb. Theory, Ser. {B}}, 53(1):1--4, 1991.

\bibitem{KRATOCHVIL1994289}
J.~Kratochv\'{i}l and J.~Matou\v{s}ek.
\newblock Intersection graphs of segments.
\newblock {\em Journal of Combinatorial Theory, Series B}, 62(2):289 -- 315,
  1994.

\bibitem{Lee16}
J.~R. Lee.
\newblock Separators in region intersection graphs.
\newblock {\em CoRR}, abs/1608.01612, 2016.

\bibitem{DBLP:journals/siamcomp/LiptonT80}
R.~J. Lipton and R.~E. Tarjan.
\newblock Applications of a planar separator theorem.
\newblock {\em {SIAM} J. Comput.}, 9(3):615--627, 1980.

\bibitem{Marx06}
D.~Marx.
\newblock Parameterized complexity of independence and domination on geometric
  graphs.
\newblock In {\em Parameterized and Exact Computation, Second International
  Workshop, {IWPEC} 2006, Z{\"{u}}rich, Switzerland, September 13-15, 2006,
  Proceedings}, pages 154--165, 2006.

\bibitem{MarxP15}
D.~Marx and M.~Pilipczuk.
\newblock Optimal parameterized algorithms for planar facility location
  problems using voronoi diagrams.
\newblock In N.~Bansal and I.~Finocchi, editors, {\em ESA 2015 Proc.}, volume
  9294 of {\em LNCS}, pages 865--877. Springer, 2015.

\bibitem{MarxP15a}
D.~Marx and M.~Pilipczuk.
\newblock Optimal parameterized algorithms for planar facility location
  problems using {V}oronoi diagrams.
\newblock {\em CoRR}, abs/1504.05476, 2015.

\bibitem{DBLP:journals/corr/Matousek14}
J.~Matou{\v{s}}ek.
\newblock Intersection graphs of segments and $\exists\mathbb{R}$.
\newblock {\em CoRR}, abs/1406.2636, 2014.

\bibitem{DBLP:journals/cpc/Matousek14}
J.~Matou{\v{s}}ek.
\newblock Near-optimal separators in string graphs.
\newblock {\em Combinatorics, Probability {\&} Computing}, 23(1):135--139,
  2014.

\bibitem{DBLP:journals/jct/McDiarmidM13}
C.~McDiarmid and T.~M{\"{u}}ller.
\newblock Integer realizations of disk and segment graphs.
\newblock {\em J. Comb. Theory, Ser. {B}}, 103(1):114--143, 2013.

\bibitem{DBLP:journals/jct/RobertsonST94}
N.~Robertson, P.~D. Seymour, and R.~Thomas.
\newblock Quickly excluding a planar graph.
\newblock {\em J. Comb. Theory, Ser. {B}}, 62(2):323--348, 1994.

\bibitem{DBLP:journals/jcss/SchaeferSS03}
M.~Schaefer, E.~Sedgwick, and D.~\v{S}tefankovi\v{c}.
\newblock Recognizing string graphs in {NP}.
\newblock {\em J. Comput. Syst. Sci.}, 67(2):365--380, 2003.

\bibitem{Schaefer2017}
M.~Schaefer and D.~{\v{S}}tefankovi{\v{c}}.
\newblock Fixed points, Nash equilibria, and the existential theory of the
  reals.
\newblock {\em Theory of Computing Systems}, 60(2):172--193, Feb 2017.


\bibitem{Scheinerman}
E.~Scheinerman.
\newblock {\em Intersection classes and multiple intersection parameters of
  graphs}.
\newblock PhD thesis, Princeton University, 1984.

\bibitem{SmithW98}
W.~D. Smith and N.~C. Wormald.
\newblock Geometric separator theorems and applications.
\newblock In {\em FOCS 1998 Proc.}, pages 232--243, Washington, DC, USA, 1998.
  IEEE Computer Society.
  
  \bibitem{Tovey84}
C.~A. Tovey.
\newblock A simplified NP-complete satisfiability problem.
\newblock {\em Discrete Applied Mathematics}, 8(1):85--89, 1984.

\bibitem{ZverovichZ95}
I.~E. Zverovich and V.~E. Zverovich.
\newblock An induced subgraph characterization of domination perfect graphs.
\newblock {\em Journal of Graph Theory}, 20(3):375--395, 1995.

\end{thebibliography}

\end{document}